\RequirePackage[l2tabu,orthodox]{nag}
\documentclass
[11pt,letterpaper]
{article} 

\usepackage[notes=true,later=false,camera=true]{dtrt}
\usepackage[utf8]{inputenc}
\usepackage{etex}
\usepackage{ stmaryrd }
\usepackage{xspace,enumerate}
\usepackage[T1]{fontenc}
\usepackage[full]{textcomp}
\usepackage[american]{babel}
\usepackage{mathtools}

\usepackage{amsthm}
\usepackage{empheq}

   \usepackage{hyperref}
   \hypersetup{hyperindex=true,pdfpagemode=UseOutlines,bookmarksnumbered=true,bookmarksopen=true,bookmarksopenlevel=2,pdfstartview=FitH,pdfborder={0 0 1},linkbordercolor=dt@linkcolor,citebordercolor=dt@linkcolor,urlbordercolor=dt@linkcolor}
\usepackage[capitalise,nameinlink]{cleveref}
\crefname{lemma}{Lemma}{Lemmas}
\crefname{fact}{Fact}{Facts}
\newcommand{\colorconstraints}{\text{Color Constraints}}
\crefname{colorconstraints}{(color constraints)}{Color Constraints}
\crefformat{colorconstraints}{#2\colorconstraints#3}
\crefname{indsetconstraints}{(indset constraints)}{IndSet Constraints}
\crefformat{indsetconstraints}{#2$\mathsf{IndSet\ Axioms}$#3}
\crefname{theorem}{Theorem}{Theorems}
\crefname{mtheorem}{Theorem}{Theorems}
\crefname{corollary}{Corollary}{Corollaries}
\crefname{claim}{Claim}{Claims}
\crefname{example}{Example}{Examples}
\crefname{algorithm}{Algorithm}{Algorithms}
\crefname{problem}{Problem}{Problems}
\crefname{definition}{Definition}{Definitions}
\usepackage{paralist}
\usepackage{turnstile}
\usepackage{mdframed}
\usepackage{tikz}
\usepackage{caption}
\DeclareCaptionType{Algorithm}
\usepackage{newfloat}
\newtheorem{theorem}{Theorem}[section]
\newtheorem{mtheorem}{Theorem}
\newtheorem*{theorem*}{Theorem}

\newtheorem*{proposition*}{Proposition}
\newtheorem{lemma}[theorem]{Lemma}
\newtheorem*{lemma*}{Lemma}
\newtheorem{corollary}[theorem]{Corollary}
\newtheorem*{conjecture*}{Conjecture}
\newtheorem{fact}[theorem]{Fact}
\newtheorem*{fact*}{Fact}

\newtheorem*{hypothesis*}{Hypothesis}

\theoremstyle{definition}
\newtheorem{definition}[theorem]{Definition}
\newtheorem*{definition*}{Definition}

\newtheorem{algorithm}[theorem]{Algorithm}

\theoremstyle{remark}
\newtheorem{claim}[theorem]{Claim}
\newtheorem*{claim*}{Claim}
\newtheorem{remark}[theorem]{Remark}
\newtheorem*{remark*}{Remark}

\newtheorem*{observation*}{Observation}

\usepackage[
letterpaper,
top=1in,
bottom=1.1in,
left=1in,
right=1in]{geometry}
\usepackage{newpxtext} 
\usepackage{textcomp} 
\usepackage[varg,bigdelims]{newpxmath}
\usepackage[scr=rsfso]{mathalfa}
\usepackage{bm} 
\let\mathbb\varmathbb
\usepackage{microtype}

\allowdisplaybreaks
\newcommand{\FormatAuthor}[3]{
\begin{tabular}{c}
#1 \\ {\small\texttt{#2}} \\ {\small #3}
\end{tabular}
}

\newcommand{\cC}{\mathsf{C}}
\newcommand{\Dec}{\mathrm{Dec}}
\newcommand{\R}{{\mathbb R}}
\newcommand{\N}{{\mathbb N}}
\newcommand{\norm}[1]{\lVert #1 \rVert}

\newcommand{\abs}[1]{\lvert #1 \rvert}

\newcommand{\eps}{\varepsilon}
\newcommand{\defeq}{\coloneqq}
\newcommand{\F}{{\mathbb F}}

\newcommand{\E}{{\mathbb E}}
\newcommand{\1}{\mathbf{1}}

\newcommand{\Bits}{\{0,1\}}

\newcommand{\Fits}{\{-1,1\}}

\newcommand{\cH}{\mathcal H}

\newcommand{\val}{\mathrm{val}}

\newcommand{\mcom}{\,,}

\newcommand{\kikuchi}{Kikuchi }

\newcommand{\polylog}{\mathrm{polylog}}

\let\svthefootnote\thefootnote
\newcommand\blfootnote[1]{%
  \let\thefootnote\relax%
  \footnotetext{#1}%
  \let\thefootnote\svthefootnote%
}

\newcommand{\inner}[2]{\langle #1, #2 \rangle}

\newcommand{\mc}[1]{\mathcal{#1}}

\begin{document}

\title{A Near-Cubic Lower Bound for $3$-Query Locally Decodable Codes from Semirandom CSP Refutation}

\author{
\begin{tabular}[h!]{ccc}
     \FormatAuthor{Omar Alrabiah\thanks{Supported in part by a Saudi Arabian Cultural Mission (SACM) Scholarship, NSF CCF-2228287 and V.\ Guruswami's Simons Investigator Award.}}{oalrabiah@berkeley.edu}{UC Berkeley}
   \FormatAuthor{Venkatesan Guruswami\thanks{Supported in part by NSF grants CCF-2228287 and CCF-2211972 and a Simons Investigator award.}}{venkatg@berkeley.edu}{UC Berkeley}
         \\ \\
         \FormatAuthor{Pravesh K.\ Kothari\thanks{Supported in part by an NSF CAREER Award \#2047933, a Google Research Scholar Award, and a Sloan Fellowship.}}{praveshk@cs.cmu.edu}{Carnegie Mellon University}
               \FormatAuthor{Peter Manohar\thanks{Supported in part by an ARCS Scholarship, NSF Graduate Research Fellowship (under grant numbers DGE1745016 and DGE2140739), and NSF CCF-1814603.}}{pmanohar@cs.cmu.edu}{Carnegie Mellon University}
\end{tabular}
} %
\date{}

\maketitle\blfootnote{Any opinions, findings, and conclusions or recommendations expressed in this material are those of the author(s) and do not necessarily reflect the views of the National Science Foundation.}
\thispagestyle{empty}

\begin{abstract}
A code $\cC \colon \Bits^k \to \Bits^n$ is a $q$-locally decodable code ($q$-LDC) if one can recover any chosen bit $b_i$ of the message $b \in \Bits^k$ with good confidence by randomly querying the encoding $x \coloneqq \cC(b)$ on at most $q$ coordinates. Existing constructions of $2$-LDCs achieve $n = \exp(O(k))$, and lower bounds show that this is in fact tight. However, when $q = 3$, far less is known: the best constructions achieve $n = \exp(k^{o(1)})$, while the best known results only show a quadratic lower bound $n \geq \tilde{\Omega}(k^2)$ on the blocklength.

In this paper, we prove a near-cubic lower bound of $n \geq \tilde{\Omega}(k^3)$ on the blocklength of $3$-query LDCs. This improves on the best known prior works by a \emph{polynomial} factor in $k$. Our proof relies on a new connection between LDCs and refuting constraint satisfaction problems with limited randomness. Our quantitative improvement builds on the new techniques for refuting \emph{semirandom} instances of CSPs developed in \cite{GuruswamiKM22,HsiehKM23} and, in particular, relies on bounding the spectral norm of appropriate \emph{Kikuchi} matrices.
\end{abstract}

\clearpage
 \microtypesetup{protrusion=false}
  \tableofcontents{}
  \microtypesetup{protrusion=true}

\thispagestyle{empty}
\clearpage

\pagestyle{plain}
\setcounter{page}{1}

\section{Introduction}
\label{sec:intro}

A binary \emph{locally decodable code} (LDC) $\cC \colon \Bits^k \to \Bits^n$ maps a $k$-bit message $b \in \Bits^k$ to an $n$-bit codeword $x \in \Bits^n$ with the property that the receiver, when given oracle access to $y \in \Bits^n$ obtained by corrupting $x$ in a constant fraction of coordinates, can recover any chosen bit $b_i$ of the original message with good confidence by only querying $y$ in a few locations. More formally, a code $\cC$ is $q$-locally decodable if for any input $i \in [k]$, the decoding algorithm makes at most $q$ queries to the corrupted codeword $y$ and recovers the bit $b_i$ with probability $1/2 + \eps$, provided that $\Delta(y, \cC(b)) \coloneqq \abs{\{v \in [n] : y_v \ne \cC(b)_v\}} \leq \delta n$, where $\delta,\eps$ are constants. Though formalized later in \cite{KT00}, locally decodable codes were instrumental in the proof of the PCP theorem \cite{AS98, ALMSS98}, and have deep connections to many other areas of complexity theory (see Section~7 in~\cite{Yek12}), including worst-case to average-case reductions \cite{Tre04}, private information retrieval \cite{Yekhanin10}, secure multiparty computation \cite{IshaiK04}, derandomization \cite{DvirS05}, matrix rigidity \cite{Dvir10}, data structures \cite{Wolf09,ChenGW10}, and fault-tolerant computation \cite{Romashchenko06}.

A central research focus in coding theory is to understand the largest possible \emph{rate} achievable by a $q$-query locally decodable code. For the simplest non-trivial setting of $q=2$ queries, we have a complete understanding: the Hadamard code provides an LDC with a blocklength $n = 2^k$ and an essentially matching lower bound of $n = 2^{\Omega(k)}$ was shown in~\cite{KdW04,GKST06,Bri16,Gop18}.

In contrast, there is a wide gap in our understanding of $3$ or higher query LDCs.
The best known constructions are based on families of \emph{matching vector codes}~\cite{Yek08, Efremenko09, DGY11} and achieve $n = 2^{k^{o(1)}}$. In particular, the blocklength  is slightly subexponential in $k$ and asymptotically improves on the rate achievable by $2$-query LDCs. The known lower bounds, on the other hand, are far from this bound. The first LDC lower bounds are due to Katz and Trevisan~\cite{KT00}, who proved that $q$-query LDCs require a blocklength of $n \geq \Omega(k^{\frac{q}{q-1}})$. This was later improved in 2004 by Kerenedis and de Wolf~\cite{KdW04} via a ``quantum argument'' to obtain $n \geq k^{\frac{q}{q-2}}/\polylog(k)$ when $q$ is even, and $n \geq k^{\frac{q+1}{q-1}}/\polylog(k)$ when $q$ is odd. For the first nontrivial setting of $q=3$, their result yields a nearly quadratic lower bound of $n \geq \Omega(k^2/\log^2 k)$ on the blocklength. Subsequently, Woodruff~\cite{Woo07, Woo12} improved this bound by $\polylog(k)$ factors to obtain a lower bound of $n \geq \Omega(k^2/\log k)$ for non-linear codes, and $n \geq \Omega(k^2)$ for linear codes. Very recently, Bhattacharya, Chandran, and Ghoshal~\cite{BCG20} used a combinatorial method to give a new proof of the quadratic lower bound of $n \geq \Omega(k^2/\log k)$, albeit with a few additional assumptions on the code.

\parhead{Our Work.} In this work, we show a near-cubic lower bound $n \geq k^3/\polylog(k)$ on the blocklength of any $3$-query LDC. This improves on the previous best lower bound by a $\tilde{O}(k)$ factor. More precisely, we prove:
\begin{mtheorem}\label{mthm:main}
Let $\cC \colon \Bits^k \to \Bits^n$ be a code that is $(3, \delta, \eps)$-locally decodable. Then, it must hold that $k^3 \leq n \cdot O((\log^{6} n)/\eps^{32} \delta^{16})$. In particular, if $\delta, \eps$ are constants, then $n \geq \Omega(k^3/\log^6 k)$.
\end{mtheorem}

We have not attempted to optimize the dependence on $\eps$ and $\delta$ in \cref{mthm:main}; for the specific case of binary \emph{linear} codes, one can obtain slightly better dependencies on $\log k, \eps, \delta$, as we show in \cref{thm:linreduction,cor:linlb}. It is straightforward to extend \cref{mthm:main} to nonbinary alphabets with a polynomial loss in the alphabet size, and we do so in \cref{thm:main-gen-alpha} in \cref{sec:general-alphabets}. Finally, using known relationships between locally correctable codes (LCCs) and LDCs (e.g., Theorem A.6 of~\cite{BGT17}), \cref{mthm:main} implies a similar lower bound for $3$-query LCCs.

Our main tool is a new connection between the existence of locally decodable codes and refutation of instances of Boolean CSPs with limited randomness. This connection is similar in spirit to the connection between PCPs and hardness of approximation for CSPs, in which one produces a $q$-ary CSP from a PCP with a $q$-query verifier by adding, for each possible query set of the verifier, a local constraint that asserts that the verifier accepts when it queries this particular set. To refute the resulting CSP instance, our proof builds on the spectral analysis of \emph{Kikuchi matrices} employed in the recent work of~\cite{GuruswamiKM22} (and the refined argument in~\cite{HsiehKM23}), which obtained strong refutation algorithms for semirandom and smoothed CSPs and proved the hypergraph Moore bound conjectured by Feige~\cite{Fei08} up to a single logarithmic factor.

Up to $\polylog(k)$ factors, the best known lower bound of $n \geq k^{\frac{q+1}{q-1}}/\polylog(k)$ for $q$-LDCs for odd $q$ can be obtained by simply observing that a $q$-LDC is also a $(q+1)$-LDC, and then invoking the lower bound for $(q+1)$-query LDCs.
Our improvement for $q = 3$ thus comes from obtaining the same tradeoff with $q$ as in the case of even $q$, but now for $q = 3$. 
For technical reasons, our proof does not extend to odd $q \geq 5$; we briefly mention at the end of \cref{sec:proofoverview} the place where the natural generalization fails. We leave proving a lower bound of $n \geq k^{\frac{q}{q-2}}/\polylog(k)$ for all \emph{odd} $q \ge 5$ as an intriguing open problem.

\subsection{Proof overview}
\label{sec:proofoverview}
The key insight in our proof is to observe that for any $q$, a $q$-LDC yields a collection of $q$-XOR instances, one for each possible message,  and a typical instance has a high value, i.e., there's an assignment that satisfies $\frac{1}{2}+\eps$-fraction of the constraints. To prove a lower bound on the blocklength $n$ for $3$-LDCs, it is then enough to show that for any purported construction with $n \ll k^3$, the associated $3$-XOR instance corresponding to a uniformly random message has a low value. We establish such a claim by producing a refutation (i.e., a certificate of low value), building on tools from the recent work on refuting smoothed instances of Boolean CSPs~\cite{GuruswamiKM22,HsiehKM23}. 

For this overview, we will assume that the code $\cC$ is a \emph{linear} $q$-LDC. We will also write the code using $\Fits$ notation, so that $\cC \colon \Fits^k \to \Fits^n$.
By standard reductions (Lemma 6.2 in \cite{Yek12}), one can assume that the LDC is in normal form: there exist $q$-uniform hypergraph matchings $\cH_1, \dots, \cH_k$, each with $\Omega(n)$ hyperedges,\footnote{A $q$-uniform hypergraph $\cH_i$ is a collection of subsets of $[n]$, called hyperedges, each of size exactly $q$. The hypergraph $\cH_i$ is a matching if all the hyperedges are disjoint.} and the decoding procedure on input $i \in [k]$ simply chooses a uniformly random $C \in \cH_i$, and outputs $\prod_{v \in C} x_v$. Because $\cC$ is linear, when $x = \cC(b)$ is the encoding of $b$, the decoding procedure recovers $b_i$ with probability $1$. In other words, for any $b \in \Fits^k$, the assignment $x = \cC(b)$ satisfies the set of $q$-XOR constraints $\forall i \in [k], C \in \cH_i, \prod_{v \in C} x_v = b_i$.

\parhead{The XOR Instance.} The above connection now suggests the following approach: let $b \in \Fits^k$ be chosen randomly, and consider the $q$-XOR instance with constraints $\forall i \in [k], C \in \cH_i, \prod_{v \in C} x_v = b_i$. Since $\cC$ is a linear $q$-LDC, this set of constraints will be satisfiable for every choice of $b$. Thus, proving that the instance is unsatisfiable, with high probability for a uniformly random $b$, implies a contradiction.

One might expect to show unsatisfiability of a $q$-XOR instance produced by a sufficiently random generation process by using natural probabilistic arguments. Indeed, if the instance was ``fully random'' (i.e., both $\cH_i$'s and $b_i$'s chosen uniformly at random from their domain), or even semirandom (where $\cH_i$'s are worst-case but each constraint $C$ has a uniformly random ``right hand side'' $b_C \in \Fits$), then a simple union bound argument suffices to prove unsatisfiability. 

The main challenge in our setting is that the $q$-XOR instances have significantly \emph{limited} randomness even compared to the semirandom setting: all the constraints $C \in \cH_i$ share the \emph{same} right hand side $b_i$. In particular, the $q$-XOR instance on $n$ variables has $k\ll n$ bits of independent randomness. 

We establish the unsatisfiability of such a $q$-XOR instance above by constructing a subexponential-sized SDP-based certificate of low value. A priori, bounding the SDP value might seem like a rather roundabout route to show unsatisfiability of a $q$-XOR instance.  However, shifting to this stronger target allows us to leverage the techniques introduced in the recent work of~\cite{GuruswamiKM22} on \emph{semirandom} CSP refutation and to show existence of such certificates of unsatisfiability. Despite the significantly smaller amount of randomness in the $q$-XOR instances produced in our setting, compared to, e.g., semirandom instances, we show that an appropriate adaptation of the techniques from~\cite{GuruswamiKM22} is powerful enough to exploit the combinatorial structure in our instances and succeed in refuting them.

\parhead{Warmup: the case when $q$ is even.}
Certifying unsatisfiability of $q$-XOR instances when $q$ is even is known to be, from a technical standpoint, substantially easier compared to the case when $q$ is odd. 
As a warmup, we will first sketch a proof of the known lower bound for $q$-LDCs when $q$ is even, via our CSP refutation approach. A full formal proof is presented in \cref{sec:even-q}. 

The refutation certificate is as follows. Let $\ell$ be a parameter to be chosen later, and let $N \coloneqq {n \choose \ell}$. For a set $C \in {[n] \choose q}$,\footnote{We use ${[n] \choose t}$ to denote the collection of subsets of $[n]$ of size exactly $t$.} we let $A^{(C)} \in \R^{N \times N}$ be the matrix indexed by sets $S \in {[n] \choose \ell}$, where $A^{(C)}(S, T) = 1$ if $S \oplus T = C$, and $0$ otherwise, where $S \oplus T$ denotes the symmetric difference of $S$ and $T$. We note that $S \oplus T = C$ if and only if $S = C_1 \cup Q$ and $T = C_2 \cup Q$, where $C_1$ is half of the clause $C$, $C_2$ is the other half of the clause $C$, and $Q$ is an arbitrary subset of $[n] \setminus C$ of size $\ell - q/2$. This matrix $A^{(C)}$ is the Kikuchi matrix (also called symmetric difference matrix) of \cite{WeinAM19}. We then set $A = \sum_{i = 1}^k b_i \sum_{C \in \cH_i} A^{(C)}$. By looking at the quadratic form $y^{\top} A y$ where $y$ is defined as $y_S \coloneqq \prod_{v \in S} x_v$, where $x = \cC(b)$, it is simple to observe that $\norm{A}_2 \geq (\ell/n)^{q/2} \cdot \sum_{i = 1}^k \abs{\cH_i} \geq (\ell/n)^{q/2} \Omega(k n)$, and this holds regardless of the draw of $b \gets \Fits^k$.

As each $b_i$ is an independent bit from $\Fits$, the matrix $A$ is the sum of $k$ independent, mean~$0$ random matrices: we can write $A = \sum_{i = 1}^k b_i A_i$, where $A_i \coloneqq \sum_{C \in \cH_i} A^{(C)}$. We can then bound $\norm{A}_2$ using Matrix Khintchine, which implies that $\norm{A}_2 \leq O(\Delta)(\sqrt{k \ell \log n})$ with high probability over $b$, where $\Delta$ is the maximum $\ell_1$-norm of a row in any $A_i$. One technical issue is that there are rows with abnormally large $\ell_1$-norm, so $\Delta$ can be as large as $\Omega(\ell)$. We show that when $\ell \leq n^{1 - 2/q}$, one can ``zero out'' rows of $A_i$ carefully so that each row/column has at most one nonzero entry.\footnote{Concretely, one sets $A_i(S,T) = 1$ if $S \oplus T  = C \in \cH_i$, and $\abs{S \oplus C'}, \abs{T \oplus C'} \ne \ell$ for all other $C' \in \cH_i \setminus C$. In other words, one sets $A_i(S,T) = 1$ if $A^{(C)}(S,T) = 1$ for some $C \in \cH_i$ and the $S$-th row and $T$-th column are $0$ in $A^{(C')}$ for all other $C' \in \cH_i \setminus \{C\}$.} This allows us to set $\Delta = 1$ provided that $\ell \leq n^{1 - 2/q}$.\footnote{The ``zeroing out'' step is a variant of the row pruning argument in~\cite{GuruswamiKM22}, which uses a sophisticated concentration inequality for polynomials~\cite{schudysviridenko} to show that almost all of the rows of $A_i$ have $\ell_1$-norm at most $\polylog(n)$. As shown in~\cite{HsiehKM23}, by doing this explicitly and without using concentration inequalities, we save on the $\polylog(n)$ factor.}

Combining, we thus have that for $\ell \leq n^{1 - 2/q}$,
\begin{flalign*}
&(\ell/n)^{q/2} \Omega(k n) \leq \norm{A}_2 \leq O(\sqrt{k \ell \log n}) \enspace.
\end{flalign*}
Taking $\ell = n^{1 - 2/q}$ to be the largest possible setting of $\ell$ for which the above holds, we obtain the desired lower bound of $k \leq n^{1 - 2/q} \cdot \polylog(n)$.

\parhead{The case of $q = 3$.}
When $q = 3$, or more generally when $q$ is odd, the matrices $A^{(C)}$ are no longer meaningful, as the condition $S \oplus T = C$ is never satisfied. A naive attempt to salvage the above approach is to simply allow the columns of $A^{(C)}$ to be indexed by sets of size $\ell + 1$, rather than $\ell$. However, this asymmetry in the matrix causes the spectral certificate to obtain a suboptimal dependence in terms of $q$, leading to a final bound of $k \leq n^{1 - 2/(q+1)} \polylog(n)$, the same as the current state-of-the-art lower bound for odd $q$. This is precisely the issue that in general makes refuting $q$-XOR instances for odd $q$ technically more challenging than even $q$. The asymmetric matrix effectively pretends that $q$ is $q + 1$, and thus obtains the ``wrong'' dependence on $q$.

Our idea is to transform a $3$-LDC into a $4$-XOR instance and then use an appropriate Kikuchi matrix to find a refutation for the resulting $4$-XOR instance. The transformation works as follows. We randomly partition $[k]$ into two sets, $L,R$, and fix $b_j = 1$ for all $j \in R$. Then, for each \emph{intersecting pair} of constraints $C_i, C_j$ that intersect with $C_i \in \cH_i, i \in L$, $C_j \in \cH_j, j \in R$, we add the derived constraint $C_i \oplus C_j$ to our new $4$-XOR instance, with right hand side $b_i$.\footnote{If $\abs{C_i \cap C_j} = 2$, then the derived constraint is a $2$-XOR constraint, not $4$-XOR. This is a minor technical issue that can be circumvented easily, so we will ignore it for the proof overview.} Because the $3$-XOR instance was satisfiable, the $4$-XOR instance is also satisfiable. Moreover, the $4$-XOR instance has $\sim k^2 n$ constraints, as a typical $v \in [n]$ participates in $\sim k$ hyperedges in $\cup_{i = 1}^k \cH_i$, and hence can be ``canceled'' to form $k^2$ derived constraints.

The partition $(L, R)$ is a technical trick that allows us to produce $\sim k^2 n$ constraints in the $4$-XOR instance while preserving $k$ independent bits of randomness in the right hand sides of the constraints. If we considered \emph{all} derived constraints, rather than just those that cross the partition $(L,R)$, then it would be possible to produce derived constraints where the right hand sides have nontrivial correlations. Specifically, one could produce $3$ constraints with right hand sides $b_i b_j, b_j b_t, b_i b_t$, which are pairwise independent but not $3$-wise independent. With the partitioning, however, the right hand sides of any two constraints must either be equal or independent, and in particular there are no nontrivial correlations.

The fact that we have produced more constraints in the $4$-XOR instance is crucial, as otherwise we could only hope to obtain the same bound as in the $q = 4$ case in the warmup earlier. However, our reduction does not produce an instance with the same structure as a $4$-XOR instance arising from a $4$-LDC: if we let $\cH'_i$ for $i \in L$ denote the set of derived constraints with right hand side $b_i$, then we clearly can see that $\cH'_i$ is not a matching. In fact, the typical size of $\cH'_i$ is $\Omega(nk)$, whereas a matching can have at most $n/q$ hyperedges.

Nonetheless, we can still apply the CSP refutation machinery to try to refute this $4$-XOR instance. However, because each $\cH'_i$ is no longer a matching, the ``zeroing out'' step now only works if we assume that any pair $p = (u,v)$ of vertices appears in at most $\polylog(n)$ hyperedges in the original $3$-uniform hypergraph $\cup_{i = 1}^k \cH_i$. But, if we make this assumption, the rest of the proof follows the blueprint of the even $q$ case, and we can prove that $n \geq k^3/\polylog(k)$. We note that a recent work~\cite{BCG20} managed to reprove that $n \geq k^2/\polylog(k)$ under a similar assumption about pairs of vertices.

Thus, the final step of the proof is to remove the assumption by showing that no pair of vertices can appear in too many hyperedges. Suppose that we do have many ``heavy'' pairs $p = (u,v)$ that appear in $\gg \log n$ clauses in the original $3$-uniform hypergraph $\cH \coloneqq \cup_{i = 1}^k \cH_i$. Now, we transform the $3$-XOR instance into a bipartite $2$-XOR instance (\cite{AbascalGK21, GuruswamiKM22}) by replacing each heavy pair $p$ with a new variable $y_p$. That is, the $3$-XOR clause $C = (u,v,w)$ in $\cH_i$ now becomes the $2$-XOR clause $(p, w)$, where $p$ is a new variable. In other words, the constraint $x_u x_v x_w = b_i$ is replaced by $y_p x_w = b_i$. Each clause in the bipartite $2$-XOR instance now uses one variable from the set of heavy pairs, and one from the original set of variables $[n]$. We then show that if there are too many heavy pairs, then this instance has a sufficient number of constraints in order to be refuted, and is thus not satisfiable, which is again a contradiction. 

Finally, we note that for larger odd $q \geq 5$, the proof showing that there not too many heavy pairs breaks down, and this is what prevents us from generalizing \cref{mthm:main} to all odd $q$. 
\subsection{Discussion: LDCs and the CSP perspective}
Prior work on lower bounds for $q$-LDCs reduce $q$-query LDCs with even $q$ to $2$-query LDCs, and then apply the essentially tight known lower bounds for $2$-query LDCs. (To handle the odd $q$ case, they essentially observe that a $q$-LDC is also a $(q + 1)$-LDC.) While the warmup proof we sketched earlier (and present in \cref{sec:even-q}) for even $q$ is in the language of CSP refutation, it is in fact very similar to the reduction from $q$-LDCs to $2$-LDCs for $q$ even used in the proof in \cite{KdW04}. The reduction in \cite{KdW04} (see also Exercise~4 in~\cite{Gopi19}) employs a certain tensor product, and while it is not relevant to their argument, the natural matrix corresponding to the $2$-LDC produced by their reduction is in fact very closely related to the Kikuchi matrix $A$ of \cite{WeinAM19}. 

The main advantage of the CSP refutation viewpoint is that it suggests a natural route to analyze $q$-LDCs for \emph{odd} $q$ via an appropriately modified Kikuchi matrix. By viewing the $3$-LDC as a $3$-XOR instance, we obtain a natural way to produce a related $4$-XOR instance using a reduction that \emph{does not correspond to a $4$-LDC}. In fact, if our reduction were to only produce a $4$-LDC, then we would not expect to obtain an improved $3$-LDC lower bound without improving the $4$-LDC lower bound as well. In a sense, this relates to the key strength of the CSP viewpoint in that it is arguably the ``right'' level of abstraction. On one hand, it naturally suggests reductions from $3$-LDCs to $4$-XOR that are rather unnatural if one were to follow the more well-trodden route of reducing odd query LDCs to even query ones. On the other hand, the ideas from semirandom CSP refutation are resilient enough to apply, with some effort, to even the more general, non-semirandom instances arising in such reductions, and so we can still prove lower bounds. Further exploration of such an approach to  obtain stronger lower bounds for LDCs is an interesting research direction. 

As we remarked above, our refutation-based proof of known $q$-LDC lower bounds for even $q$ turns out to be closely related to the existing proofs~\cite{KdW04, Woo07} that establish the lower bounds via a black-box reduction to $2$-LDC lower bounds. Because of this, one might wonder if our lower bound in \cref{mthm:main} can also be proven via a black-box reduction to $2$-LDC lower bounds. This turns out to be the case but only for \emph{linear} $3$-LDCs, and we present the argument in \cref{sec:2-ldc-reduction}. 
Curiously, our reduction-based proof requires \emph{two} black-box invocations of the $2$-LDC lower bound, which is unlike the existing proofs for even $q$ that require only one invocation~\cite{KdW04, Woo07}. Moreover, our reduction-based proof does not extend to non-linear codes; we discuss the barriers in \cref{rem:linearity}.
\section{Preliminaries}
\label{sec:prelims}

\subsection{Basic notation}
We let $[n]$ denote the set $\{1, \dots, n\}$. For two subsets $S, T \subseteq [n]$, we let $S \oplus T$ denote the symmetric difference of $S$ and $T$, i.e., $S \oplus T \coloneqq \{i : (i \in S \wedge i \notin T) \vee (i \notin S \wedge i \in T)\}$. For a natural number $t \in \N$, we let ${[n] \choose t}$ be the collection of subsets of $[n]$ of size exactly $t$.

For a rectangular matrix $A \in \R^{m \times n}$, we let $\norm{A}_2 \coloneqq \max_{x \in \R^m, y \in \R^n: \norm{x}_2 = \norm{y}_2 = 1} x^{\top} A y$ denote the spectral norm of $A$.
\subsection{Locally decodable codes and hypergraphs}
\begin{definition}
A hypergraph $\cH$ with vertices $[n]$ is a collection of subsets $C \subseteq [n]$ called hyperedges. We say that a hypergraph $\cH$ is \emph{$q$-uniform} if $\abs{C}= q$ for all $C \in \cH$, and we say that $\cH$ is a \emph{matching} if all the hyperedges in $\cH$ are disjoint. For a subset $Q \subseteq [n]$, we define the degree of $Q$ in $\cH$, denoted $\deg_{\cH}(Q)$, to be $\abs{\{C \in \cH : Q \subseteq C\}}$.
\end{definition}

\begin{definition}[Locally Decodable Code]
\label{def:LDC}
A code $\cC \colon \Bits^k \to \Bits^n$ is $(q, \delta, \eps)$-locally decodable if there exists a randomized decoding algorithm $\Dec(\cdot)$ with the following properties. The algorithm $\Dec(\cdot)$ is given oracle access to some $y \in \Bits^n$, takes an $i \in [k]$ as input, and satisfies the following: \begin{inparaenum}[(1)] \item the algorithm $\Dec$ makes at most $q$ queries to the string $y$, and \item for all $b \in \Bits^k$, $i \in [k]$, and all $y \in \Bits^n$ such that $\Delta(y, \cC(b)) \leq \delta n$, $\Pr[\Dec^{y}(i) = b_i] \geq \frac{1}{2} + \eps$. Here, $\Delta(x,y)$ denotes the Hamming distance between $x$ and $y$, i.e., the number of indices $v \in [n]$ where $x_v \ne y_v$.\end{inparaenum}
\end{definition}

Following known reductions \cite{Yek12}, locally decodable codes can be reduced to the following normal form, which is more convenient to work with.
\begin{definition}[Normal LDC]
\label{def:normalLDC}
A code $\cC \colon \Fits^k \to \Fits^n$ is $(q, \delta, \eps)$-normally decodable if for each $i \in [k]$, there is a $q$-uniform hypergraph matching $\cH_i$ with at least $\delta n$ hyperedges such that for every $C \in \cH_i$, it holds that $\Pr_{b \gets \Fits^k}[b_i = \prod_{v \in C} \cC(b)_v] \geq \frac{1}{2} + \eps$.
\end{definition}

\begin{fact}[Reduction to LDC Normal Form, Lemma 6.2 in~\cite{Yek12}]\label{fact:normalform}
Let $\cC \colon \Bits^k \to \Bits^n$ be a code that is $(q, \delta, \eps)$-locally decodable. Then, there is a code $\cC' \colon \Fits^k \to \Fits^{O(n)}$ that is  $(q, \delta', \eps')$ normally decodable, with $\delta' \geq \eps \delta/3q^2 2^{q-1}$ and $\eps' \geq \eps/2^{2q}$.
\end{fact}

\subsection{The Matrix Khintchine inequality}
Our work will use the expectation form of the standard rectangular Matrix Khintchine inequality.
\begin{fact}[Rectangular Matrix Khintchine Inequality, Theorem 4.1.1 of \cite{Tropp15}]
\label{fact:matrinxkhintchine}
Let $X_1, \dots, X_k$ be fixed $d_1 \times d_2$ matrices and $b_1, \dots , b_k$ be i.i.d.\ from $\Fits$. Let $\sigma^2 \geq \max(\norm{\sum_{i = 1}^k X_i X_i^{\top}]}_2, \norm{\sum_{i = 1}^k X_i^{\top} X_i]}_2)$. Then
\begin{equation*}
\E\Bigl[\ \norm{\sum_{i = 1}^k b_i X_i}_2\ \Bigr] \leq \sqrt{2\sigma^2 \log(d_1 + d_2)} \enspace.
\end{equation*}
\end{fact}
\subsection{A fact about binomial coefficients}
We will need the following fact about the ratio of two specific binomial coefficients.
\begin{fact}
\label{fact:binomialratio}
Let $n, \ell, q$ be positive integers such that $n/2 \geq \ell \geq q$. Then, $e^{3q} (\ell/n)^{q}\geq {{n - 2q}\choose {\ell - q}}/{n \choose \ell} \geq e^{-3q} (\ell/n)^{q}$.
\end{fact}
\begin{proof}
The ratio
\begin{flalign*}
&{{n - 2q}\choose {\ell - q}}/{n \choose \ell} =   \frac{(n-2q)!}{(\ell - q)! (n - \ell -q)!} \cdot \frac{\ell! (n - \ell)!}{n!} ={{n - \ell} \choose q} {\ell \choose q}/{2 q \choose q} {n \choose 2q} \enspace.
\end{flalign*}
This implies that
\begin{flalign*}
&{{n - 2q}\choose {\ell - q}}/{n \choose \ell} \leq e^{2q} \left(\frac{n - \ell}{q}\right)^q \left(\frac{\ell}{q}\right)^{q} \cdot 2^{-q} \left(\frac{n}{2q}\right)^{-2q} \leq e^{2q} q^{-2q} 2^{-q} (2q)^{2q} \left( \frac{n - \ell}{n} \right)^q \left(\frac{\ell}{n}\right)^q \leq e^{3q} \left(\frac{\ell}{n}\right)^q \enspace,
\end{flalign*}
and that
\begin{flalign*}
&{{n - 2q}\choose {\ell - q}}/{n \choose \ell} \geq \left(\frac{n - \ell}{q}\right)^q \left(\frac{\ell}{q}\right)^q \cdot 2^{-2q} \left(\frac{e n}{2q}\right)^{-2q} = e^{-2q} \cdot \left(\frac{n - \ell}{n}\right)^q \left(\frac{\ell}{n}\right)^{q} \geq e^{-2q} 2^{-q} \left(\frac{\ell}{n}\right)^{q} \geq e^{-3q} \left(\frac{\ell}{n}\right)^{q} \enspace,
\end{flalign*}
where we use that $\ell \leq n/2$. Throughout, we use that $\left(\frac{n}{k}\right)^k \leq {n \choose k} \leq \left(\frac{en}{k}\right)^k$.
\end{proof}

\section{Lower Bound for $3$-Query Locally Decodable Codes}
\label{sec:proof}
In this section, we will prove \cref{mthm:main}, our main result.

\parhead{Setup.} By \cref{fact:normalform}, in order to show that $k^3 \leq n \cdot \frac{O(\log^{6} n)}{\eps^{32} \delta^{16}}$, it suffices for us to show that for any code $\cC \colon \Fits^k \to \Fits^n$ that is  $(3, \delta, \eps)$-normally decodable, it holds that $k^{3} \leq n \cdot \frac{O(\log^{6} n)}{\eps^{16} \delta^{16}}$. As $\cC$ is $(3, \delta, \eps)$-normally decodable, this implies that there are $3$-uniform hypergraph matchings $\cH_1, \dots, \cH_k$ satisfying the property in \cref{def:normalLDC}. Let $m \coloneqq \sum_{i = 1}^k \abs{\cH_i}$ be the total number of hyperedges in the hypergraph $\cH \coloneqq \cup_{i = 1}^k \cH_i$.

The key idea in our proof is to define a $3$-XOR instance corresponding to the decoder in \cref{def:normalLDC}. By \cref{def:normalLDC}, the $3$-XOR instance we define has a high value, i.e., there is an assignment to the variables satisfying a nontrivial fraction of the constraints. To finish the proof, we show that if $n \ll k^{3}$, then the $3$-XOR instance must have small value, which is a contradiction.

We define the relevant family of $3$-XOR instances below.
\begin{mdframed}[frametitle = {The Key $3$-XOR Instances}, frametitlealignment=\centering]
\label{box:xor}
For each $b \in \Fits^k$, we define the $3$-XOR instance $\Psi_b$, where:
\begin{enumerate}[(1)]
\item The variables are $x_1, \dots, x_n \in \Fits$,
\item The constraints are, for each $i \in [k]$ and $C \in \cH_i$, $\prod_{v \in C} x_v = b_i$.
\end{enumerate}
The value of $\Psi_b$, denoted $\val(\Psi_b)$, is the maximum fraction of constraints satisfied by any assignment $x \in \Fits^n$.

We associate an instance $\Psi_b$ with the polynomial $\psi_b(x) \coloneqq \frac{1}{m} \sum_{i = 1}^k b_i \sum_{C \in \cH_i} \prod_{v \in C} x_v$, and define $\val(\psi_b) \coloneqq \max_{x \in \Fits^n} \psi_b(x)$. We note that $\val(\Psi_b) = \frac{1}{2} + \frac{1}{2}\val(\psi_b)$.
\end{mdframed}
We first observe that \cref{def:normalLDC} immediately implies that every $3$-XOR instance in the above family (indexed by $b \in \{-1,1\}^k$) $\Psi_b$ must have a non-trivially large value. Formally, we have that
\begin{equation}
\label{eq:vallowerbound}
\E_{b \gets \Fits^k}[\val(\psi_b)] \geq \E_{b \gets \Fits^k}[\psi_b(\cC(b))] \geq 2\eps \enspace,
\end{equation}
where the first inequality is by definition of $\val(\cdot)$, and the second inequality uses \cref{def:normalLDC}, as for each constraint $C \in \cH_i$ for some $i$, the encoding $\cC(b)$ of $b$ satisfies this constraint with probability $\frac{1}{2} + \eps$ for a random $b$.

\parhead{Overview: refuting the XOR instances.}
To finish the proof, it thus suffices to argue that $\E_{b \gets \Fits^k}[\val(\psi_b)]$ is small. We will do this by using a CSP refutation algorithm inspired by~\cite{GuruswamiKM22}. Our argument proceeds in two steps:
\begin{enumerate}[(1)]
\item \textbf{Decomposition:} First, we take any pair $Q = \{u,v\}$ of vertices that appears in $\gg \log n$ of the hyperedges in $\cH \coloneqq \cup_{i = 1}^k \cH_i$, and we replace this pair with a new variable $y_Q$ in all the constraints containing this pair. This process decomposes the $3$-XOR instance into a \emph{bipartite} $2$-XOR instance (\cite{AbascalGK21,GuruswamiKM22}), and a residual $3$-XOR instance where every pair of variables appears in at most $O(\log n)$ constraints.

\item \textbf{Refutation:} We then produce a ``strong refutation'' for each of the bipartite $2$-XOR and the residual 3-XOR instances that shows that the average value of the instance over the draw of $b \sim \{-1,1\}^k$ is small. This implies that each of the two instances produced and thus the original 3-XOR instance has a small expected value and finishes the proof.
\end{enumerate}

We now formally define the decomposition process. We recall a notion of degree in hypergraphs that turns out to be useful in our argument (similar to the analysis in~\cite{GuruswamiKM22}).
\begin{definition}[Degree]
\label{def:degree}
Let $\cH$ be a $q$-uniform hypergraph on $n$ vertices, and let $Q \subseteq [n]$. The degree of $Q$, $\deg_{\cH}(Q)$, is the number of $C \in \cH$ with $Q \subseteq C$.
\end{definition}

\begin{lemma}[Hypergraph Decomposition]
\label{lem:decomp}
Let $\cH_1, \dots, \cH_k$ be $3$-uniform hypergraphs on $n$ vertices, and let $\cH \coloneqq \cup_{i = 1}^k \cH_i$. Let $d \in \N$ be a threshold. Let $P \coloneqq \{\{u,v\} : \deg_{\cH}(\{u,v\}) > d\}$. Then, there are $3$-uniform hypergraphs $\cH'_1, \dots, \cH'_k$ and bipartite graphs $G_1, \dots, G_k$, with the following properties.
\begin{enumerate}[(1)]
\item Each $G_i$ is a bipartite graph with left vertices $[n]$ and right vertices $P$.
\item Each $\cH'_i$ is a subset of $\cH_i$.
\item For each $i \in [k]$, there is a one-to-one correspondence between hyperedges $C \in \cH_i \setminus \cH'_i$ and edges $e$ in $G_i$, given by $e = (w,\{u,v\}) \mapsto C = \{u,v,w\}$.
\item Let $\cH' \coloneqq \cup_{i = 1}^k \cH'_i$. Then, for any $u \ne v \in [n]$, it holds that $\deg_{\cH'}(\{u,v\}) \leq d$.
\item If $\cH_i$ is a matching, then $\cH'_i$ and $G_i$ are also matchings.
\end{enumerate}
\end{lemma}
The proof of \cref{lem:decomp} is simple, and is given in \cref{sec:decomp}.

Given the decomposition, the two main steps in our refutation are captured in the following two lemmas, which handle the $2$-XOR and $3$-XOR instances, respectively.
\begin{lemma}[$2$-XOR refutation]
\label{lem:2xor}
Fix $n \in \N$. Let $G_1, \dots, G_k$ be bipartite matchings with left vertices $[n]$ and a right vertex set $P$ of size $\abs{P} \leq nk/d$ for some $d \in \N$. For $b \in \Fits^k$, let $g_b(x,y)$ be a homogeneous quadratic polynomial defined by 
\begin{equation*}
g_b(x,y) \coloneqq \sum_{i = 1}^k b_i \sum_{e = \{v, p\} : v \in [n], p \in P} x_v y_p\mcom
\end{equation*}
and let $\val(g_b) \coloneqq \max_{x \in \Fits^n, y \in \Fits^{P}} g_b(x,y)$. Then, $\E_{b \gets \Fits^k} [ \val(g_b) ] \leq O(nk \sqrt{(\log n)/d})$.
\end{lemma}

\begin{lemma}[$3$-XOR refutation]
\label{lem:3xor}
Let $\cH_1, \dots, \cH_k$ be $3$-uniform hypergraph matchings on $n$ vertices, and let $\cH \coloneqq \cup_{i = 1}^k \cH_i$. Suppose that for any $\{u,v\} \subseteq [n]$, $\deg_{\cH}(\{u,v\}) \leq d$. Let $f_b(x) \coloneqq \sum_{i = 1}^k b_i \sum_{C \in \cH_i} \prod_{v \in C} x_v$. Then, it holds that
\begin{equation*}
\E_{b \gets \Fits^k} [ \val(f_b) ] \leq n \sqrt{k} \cdot O(d) \cdot  (nk)^{1/8} \log^{1/4} n\enspace.
\end{equation*}
\end{lemma}
We prove \cref{lem:2xor} in \cref{sec:2xor}, and we prove \cref{lem:3xor} in \cref{sec:3xor}.

With the above ingredients, we can now finish the proof of \cref{mthm:main}. 
\begin{proof}[Proof of \cref{mthm:main}]
Applying \cref{lem:decomp} with $d = O((\log n)/\eps^2\delta^2)$ for a sufficiently large constant, we decompose the instance $\Psi_b$ into $2$-XOR and $3$-XOR subinstances.\footnote{We remark that it is possible that one (but not both!) of the $2$-XOR or $3$-XOR subinstances has very few constraints, or even no constraints at all. This is not a problem, however, as then the upper bound on the value of the instance shown in corresponding lemma (either \cref{lem:2xor} or \cref{lem:3xor}) becomes trivial.} Note that as $m \leq nk$, we will have $\abs{P} \leq m/d \leq nk/d$. We have that $m \val(\psi_b) \leq \val(f_b) + \val(g_b)$ because of the one-to-one correspondence property in \cref{lem:decomp}. We also note that $m \geq \delta n k$, as $\abs{\cH_i} \geq \delta n$ for each $i$. By \cref{lem:2xor} and by taking the constant in the choice of $d$ sufficiently large, we can ensure that $\E_{b \gets \Fits^k}[\val(g_b)] \leq \eps \delta nk/3$.
Hence, by \cref{eq:vallowerbound,lem:3xor}, we have
\begin{flalign*}
&2\eps \delta nk \leq 2\eps m \leq m \E_{b \gets \Fits^k}[\val(\psi_b)] \leq  \E_{b \gets \Fits^k}[\val(f_b) + \val(g_b)] \\
&\leq \frac{\eps \delta nk}{3} +  n \sqrt{k} \cdot O(\sqrt{\log n}/\eps \delta) \cdot  (nk)^{1/8} \log^{1/4} n \\
&\implies \eps^2 \delta^2 \sqrt{k} \leq O(\sqrt{\log n}) \cdot (nk)^{1/8} \log^{1/4} n \\
&\implies k^3 \leq n \cdot O(\log^{6} n)/\eps^{16} \delta^{16} \enspace.
\end{flalign*}
We thus conclude that $k^{3} \leq n \cdot O\left(\frac{\log^{6} n}{\eps^{16} \delta^{16}}\right)$, which finishes the proof.
\end{proof}

\subsection{Hypergraph decomposition: proof of \cref{lem:decomp}}
\label{sec:decomp}
We prove \cref{lem:decomp} by analyzing the following greedy algorithm.
\begin{mdframed}
  \begin{algorithm}
    \label{algo:reg-hypergraph-decomp}\mbox{}
    \begin{description}
    \item[Given:]
       $3$-uniform hypergraphs $\cH_1, \dots, \cH_k$.
    \item[Output:]
       $3$-uniform hypergraphs $\cH'_1, \dots, \cH'_k$ and bipartite graphs $G_1, \dots, G_k$.
           \item[Operation:]\mbox{}
    \begin{enumerate}
    	\item \textbf{Initialize:} $\cH'_i = \cH_i$ for all $i \in [k]$, $P = \{\{u,v\} : \deg_{\cH'}(\{u,v\}) > d\}$, where $\cH' = \cup_{i \in [k]} \cH'_i$.
	    \item \textbf{While $P$ is nonempty}:
			\begin{enumerate}[(1)]
    			\item Choose $p = \{u,v\} \in P$ arbitrarily.			
    			\item For each $i \in [k]$, $C \in \cH'_i$ with $p \in C$, remove $C$ from $\cH'_i$, and add the edge $(C \setminus p, p)$ to $G_i$.
			\item Recompute $P = \{\{u,v\} : \deg_{\cH'}(\{u,v\}) > d\}$.
    		\end{enumerate}
		\item Output $\cH'_1, \dots, \cH'_k$, $G_1, \dots, G_k$.
		\end{enumerate}
    \end{description}
  \end{algorithm}
  \end{mdframed}
  Indeed, properties (1), (2) and (5) in \cref{lem:decomp} trivially hold. Property (4) holds because otherwise the algorithm would not have terminated, as the set $P$ would still be nonempty. Property (3) holds because each hyperedge $C \in \cH_i$ starts in $\cH'_i$, and is either removed exactly once and added to $G_i$ as $(C \setminus p, p)$, or remains in $\cH'_i$ for the entire operation of the algorithm. This finishes the proof.

\subsection{Refuting the $2$-XOR instance: proof of \cref{lem:2xor}}
\label{sec:2xor}
We now prove \cref{lem:2xor}. We do this as follows. For each $e = \{v, p\}$, with $v \in [n]$, $p \in P$, define the matrix $A^{(e)} \in \R^{n \times P}$, where $A^{(e)}(v',p') = 1$ if $v' = v$ and $p' = p$, and $0$ otherwise. Let $A_i \coloneqq \sum_{e \in G_i} A^{(e)}$,  the bipartite adjacency matrix of $G_i$. Finally,  let $A \coloneqq \sum_{i = 1}^k b_i A_i$.

First, we observe that $\val(g_b) \leq \sqrt{n \abs{P}} \norm{A}_2$. Indeed, this is because for any $x \in \Fits^n, y \in \Fits^P$, we have $g_b(x,y) = x^{\top} A y \leq \norm{x}_2 \norm{y}_2 \norm{A}_2 = \sqrt{n \abs{P}} \norm{A}_2$. Thus, in order to bound $\E_{b \gets \Fits^k}[\val(g_b)]$, it suffices to bound $\E_{b}[\norm{A}_2]$.

We use \cref{fact:matrinxkhintchine} to bound $\E[\norm{A}_2]$. Indeed, we observe that $\norm{A_i}_2 \leq 1$ for each $i$, as each row/column of $A_i$ has at most one nonzero entry of magnitude $1$ because each $G_i$ is a matching. Thus, $\max(\norm{\sum_{i = 1}^k A_i A_i^{\top}},
\norm{\sum_{i = 1}^k A_i^{\top} A_i}) \leq k$. As the $b_i$'s are i.i.d.\ from $\Fits$, by \cref{fact:matrinxkhintchine} we have that $\E[\norm{A}_2] \leq O(\sqrt{k \log n})$. It thus follows that $\E[\val(g_b)] \leq \sqrt{n \abs{P}}O(\sqrt{k \log n}) \leq O(nk \sqrt{(\log n)/d})$.

\section{Refuting the $3$-XOR Instance: Proof of \cref{lem:3xor}} 
\label{sec:3xor}
In this section, we will omit the subscript and write $f$ instead of $f_b$. We will also let $m \coloneqq \abs{\cH} = \sum_{i = 1}^k \abs{\cH_i}$.

For a vertex $u \in [n]$ and a subset $C \in {[n] \choose 2}$, we will use the notation $(u, C)$ to denote the set $\{u\} \cup C$. We will assume that $k \leq n/c$ for some sufficiently large absolute constant $c$. This is without loss of generality, as otherwise we can partition $k$ into at most $c$ disjoint blocks of size $\leq n/c$, and refute each of these subinstances separately.

The main idea is inspired by the ``Cauchy-Schwarz'' trick in the context of refuting odd-arity XOR instances. Specifically, we will construct a $4$-XOR instance by ``canceling'' out every $x_u$ that appears in two different clauses. Concretely, include every element in $[k]$ into one of two sets $L,R$ uniformly at random. Then, for any $(u, C) \in \cH_i$ with $i \in L$ and $(u, C') \in \cH_j$ with $j \in R$, we construct the ``derived clause'' $C \oplus C'$ by XOR-ing both sides of the two constraints. We then relate the value of the  instance with such derived constraints to the original 3-XOR instance and produce a spectral refutation for the derived instance via an appropriate subexponential-sized matrix. This will show that the expected value of the derived instance, over the randomness of the $b_i$'s, is small, and complete the proof.

\parhead{Relating the derived $4$-XOR to the original $3$-XOR.} First, let $(L,R)$ be a partition of $[k]$ into two sets of equal size $k/2$. Let $f_{L,R}(x)$ be the following polynomial:
\begin{equation*}
f_{L,R}(x) \coloneqq \sum_{\substack{i \in L \\ j \in R}} \sum_{u \in [n]} \sum_{\substack{(u, C) \in \cH_i \\ (u, C') \in \cH_j}} b_i b_j  x_{C} x_{C'}\enspace,
\end{equation*}
where $x_{C}$ is defined as $\prod_{v \in C} x_v$. We note that because the $\cH_i$'s are matchings, after fixing $i$, $j$, and $u$, there is at most one pair $(C, C')$ in the inner sum. Informally speaking, only working with clauses derived across the partition allows us to ``preserve'' $\sim k$ independent bits of randomness in the right hand sides of the $4$-XOR instance while eliminating nontrivial correlations.  This is crucial in eventually applying the Matrix Khintchine inequality to produce a spectral refutation.

The following lemma relates $\val(f_{L,R})$ to $\val(f)$.
\begin{lemma}[Cauchy-Schwarz Trick]
\label{lem:cauchyschwarz}
Let $f$ be as in \cref{lem:3xor} and let $L,R \subseteq [k]$ be constructed by including every element in $[k]$ to be in $L$ with probability $1/2$ independently and defining $R= [k] \setminus L$. Then, it holds that $9 \cdot \val(f)^2 \leq 3n m + 4 n \E_{(L,R)} \val(f_{L,R})$. In particular, $\E_{b \in \Fits^k} [9 \cdot \val(f)^2 ] \leq 3nm + 4 n \E_{(L,R)} \E_{b \in \Fits^k} [\val(f_{L,R})]$.
\end{lemma}
\begin{proof}
Fix any assignment to $x \in \Fits^n$. We have that
\begin{flalign*}
&(3f(x))^2 = \left(\sum_{u \in [n]} x_u \sum_{i \in [k]} \sum_{(u, C) \in \cH_i}  b_i x_C\right)^2 \leq \left(\sum_{u \in [n]} x_u^2\right) \left( \sum_{u \in [n]} \left(\sum_{i \in [k]} \sum_{(u, C) \in \cH_i} b_i x_C\right)^2 \right) \\
&= n \sum_{u \in [n]} \sum_{i, j \in [k]} \sum_{\substack{(u, C) \in \cH_i \\ (u, C') \in \cH_j}} b_i b_j x_{C} x_{C'} = n\left(3\sum_{i \in [k]} \abs{\cH_i} + \sum_{u \in [n]}  \sum_{i, j \in [k], i \ne j}\sum_{\substack{(u, C) \in \cH_i \\ (u, C') \in \cH_j}} b_i b_j x_{C} x_{C'} \right) \\
&= 3nm + 4 n \cdot \E_{(L, R)} f_{L,R}(x) \enspace,
\end{flalign*}
where the first equality is because there are $3$ ways to decompose a set $C_i \in \cH_i$ with $\abs{C_i} = 3$ into a pair $(u, C)$, the inequality follows by the Cauchy-Schwarz inequality, and the last equality follows because for a pair of hypergraphs $\cH_i$ and $\cH_j$, we have $i \in L$ and $j \in R$ with probability $1/4$. Finally, $\max_{x \in \{-1,1\}^n} \E_{(L,R)} f_{L,R}(x) \leq \E_{(L,R)} \max_{x \in \{-1,1\}^n} f_{L,R}(x) = \E_{(L,R)} \val(f_{L,R})$. Thus, we have that $9 \cdot \val(f)^2 \leq 3n m + 4 n \cdot  \E_{(L,R)} \val(f_{L,R})$. 
\end{proof}

\subsection{Bounding $\val(f_{L,R})$ using CSP refutation}
It remains to bound $\E_{b \in \Fits^k} \val(f_{L,R})$ for each choice of partition $(L,R)$. We will do this by introducing a matrix $A$ for each $b \in \Fits^k$ and partition $(L,R)$, and then we will relate $\val_{f_{L,R}}$ to $\norm{A}_2$. Note that $A$ will depend on the choice of $b$ and the partition $(L,R)$. Then, we will bound $\E_{b \in \Fits^k}[\norm{A}_2]$.

To define the matrix $A$, we introduce the following definitions.
\begin{definition}
Let $u \in [n]$ be a vertex. We let $u^{(1)}$ and $u^{(2)}$ denote the elements $(u,1)$ and $(u,2)$ of $[n] \times [2]$, i.e., if we think of $[n] \times [2]$ as two copies of $[n]$, then $u^{(1)}$ is the first copy and $u^{(2)}$ is the second one. We use similar notation for sets, so if $C \subseteq [n]$, then $C^{(1)}$ and $C^{(2)}$ denote the subsets of $[n] \times [2]$ defined as $C^{(b)} = \{(i, b) : i \in C\}$ for $b \in [2]$.
\end{definition}

\begin{definition}[Half clauses]
\label{def:half-clause}
For $i \in L, j \in R$, we define the set $P_{i,j}$ of ``half clauses'' to consist of all pairs $(v^{(1)}, w^{(2)})$ such that there exist clauses $(u, C) \in \cH_i$, $(u, C') \in \cH_j$ where $v \in C$ and $w \in C'$.

We let $P_i \coloneqq \cup_{j \in R} P_{i,j}$.
\end{definition}

Our matrix is easiest to define in two steps. We first define a matrix $B$. Then, we will specify some modifications to $B$ that yield the final matrix $A$.
\begin{definition}[Our initial \kikuchi matrix] \label{def:kikuchi-matrix}
Let $\ell \coloneqq (\sqrt{n/k})/c$ for some sufficiently large constant $c$,\footnote{We note that the matrix is only well-defined if $\ell \geq 2$, but this holds because we assumed that $k \leq n/c'$ for some sufficiently large absolute constant $c'$. This is the only place where we will use this assumption.} and let $N \coloneqq {2n \choose \ell}$. For any two sets $S,T \subseteq [n] \times [2]$ and sets $C, C' \in {[n] \choose 2}$, we say that $S \overset{C,C'}{\leftrightarrow} T$ if 
\begin{enumerate}
    \item $S \oplus T = C^{(1)} \oplus C'^{(2)}$, 
    \item $\abs{S \cap C^{(1)}} = \abs{S \cap C'^{(2)}} = \abs{T \cap C^{(1)}} = \abs{T \cap C'^{(2)}} = 1$.
\end{enumerate}
Note that $C^{(1)} \oplus C'^{(2)} = C^{(1)} \cup C'^{(2)}$, as $C^{(1)}$ and $C'^{(2)}$ are disjoint by construction.

For each $i \in L$ and $C, C' \in {[n] \choose 2}$, define the $N \times N$ matrix $B^{(i, C, C')}$, indexed by sets $S \subseteq [n] \times [2]$ of size $\ell$, by setting $B^{(i, C, C')}(S,T) = 1$ if
\begin{inparaenum}[(1)]
\item $S \overset{C,C'}{\leftrightarrow} T$, and
\item each of $S$ and $T$ contains at most one half clause from $P_i$.
\end{inparaenum}
Otherwise, we set $B^{(i, C, C')}(S,T) = 0$.

Finally, we let
\begin{equation*}
B_{i,j} \coloneqq \sum_{u \in [n]} \sum_{(u, C) \in \cH_i, (u, C') \in \cH_j} B^{(i, C, C')}, \quad B_i \coloneqq \sum_{j \in R} b_j B_{i,j}, \quad \text{and} \quad B \coloneqq \sum_{i \in L} b_i B_i \enspace. 
\end{equation*}
\end{definition}

\smallskip

We note that the matrices $B_i$ in \cref{def:kikuchi-matrix} directly give a reduction from the $3$-XOR instance $f$ to a $2$-LDC, and this can be used to obtain our $3$-LDC lower bound in the specific case of \emph{linear} codes (see the proof of \cref{thm:linreduction} in \cref{sec:2-ldc-reduction}).

\begin{remark}
\label{rem:edge-count}
For a fixed choice of $(u,C) \in \cH_i$, $(u, C') \in \cH_j$ with $j \in R$, the matrix $B^{(i, C, C')}$ has exactly $4 {2n - 4 \choose \ell - 2}$ nonzero entries, \emph{if we ignore} the additional condition that $S$ and $T$ each contain at most one half clause from $P_i$. Indeed, this is because $S \overset{C,C'}{\leftrightarrow} T$ if and only if $S$ and $T$ each contain one entry of $C$ and $C'$ ($2$ choices per clause), and the remaining part of $S$ and $T$ is the same set $Q \subseteq [n] \times [2] \setminus (C^{(1)} \oplus C'^{(2)})$ of size $\ell - 2$ (which has ${2n - 4 \choose \ell - 2}$ choices).

We note that this fact is the reason for using subsets of $[n] \times [2]$ rather than just $[n]$. If we used subsets of $[n]$ only, the number of nonzero entries in $B^{(i, C, C')}$ would depend on $\abs{C \oplus C'}$, whereas with subsets of $[n] \times [2]$ we always have $\abs{C^{(1)} \oplus C'^{(2)}} = 4$.
\end{remark}

Observe that if $S \overset{C,C'}{\leftrightarrow} T$, then $S$ and $T$ each contain at least one half clause from $P_i$, namely coming from $(C, C')$. Thus, the additional condition on $S$ and $T$ is that they contain \emph{no other} half clauses. As we shall show below, this additional condition implies that $B_i$ has at most $2d$ nonzero entries per row and thus $\norm{B_i}_2 \leq 2d $, where $d$ is the parameter in the statement of \cref{lem:3xor}, \emph{without} meaningfully affecting the number of nonzero entries in each of the $B^{(i, C, C')}$'s. We note that without this condition, one can show that $\norm{B_i}_2 \geq \Omega(\ell)$, which is large.
\begin{lemma}[Nonzero entry bound]
\label{lem:degbound}
For $i \in L$, let $B_i$ be defined as in \cref{def:kikuchi-matrix}. Then, $B_i$ has at most $2d$ nonzero entries per row/column.
\end{lemma}
We postpone the proof of \cref{lem:degbound} to \cref{sec:specbound}, and now continue with the proof.

The following lemma shows that the number of nonzero entries in $B^{(i, C, C')}$ is at least $2{2n - 4 \choose \ell - 2}$, i.e., half of $4 {2n - 4 \choose \ell - 2}$; thus, the additional condition only decreases the number of nonzero entries by a factor of $2$ per derived constraint. The factor of $2$ is not important and is chosen for convenience, and determines the constant $c$ in the parameter $\ell$.

\begin{lemma}[Counting nonzero entries]
\label{lem:rowpruning}
 For some $(u,C) \in \cH_i$ and $(u, C') \in \cH_j$ with $j \in R$, let $B^{(i, C, C')}$ be as in \cref{def:kikuchi-matrix}. Then, the number of nonzero entries in $B^{(i, C, C')}$ is at least $2{2n - 4 \choose \ell - 2}$.
\end{lemma}
We postpone the proof of \cref{lem:rowpruning} to \cref{sec:rowpruning}, and now continue with the proof.

We obtain the final matrix $A$ by, for each $A^{(i, C, C')}$, zero-ing out entries of $B^{(i, C, C')}$ until it has \emph{exactly} $2{2n - 4 \choose \ell - 2}$ nonzero entries. This is identical to the ``equalizing step'' of the edge deletion process in~\cite{HsiehKM23}.
\begin{definition}[Our final Kikuchi matrix]
\label{def:matrixfinal}
For each $i \in L$ and each pair of clauses $(u,C) \in \cH_i$ and $(u,C') \in \cH_j$ with $j \in R$, let $A^{(i, C, C')}$ be the matrix obtained from $B^{(i, C, C')}$ by arbitrarily zero-ing out entries of $B^{(i, C, C')}$ until the resulting matrix has exactly $D \coloneqq 2{2n - 4 \choose \ell - 2}$ nonzero entries.

We let
\begin{equation*}
A_{i,j} \coloneqq \sum_{u \in [n]} \sum_{(u, C) \in \cH_i, (u, C') \in \cH_j} A^{(i, C, C')}, \quad A_i \coloneqq \sum_{j \in R} b_j A_{i,j}, \quad \text{and} \quad A \coloneqq \sum_{i \in L} b_i A_i \enspace.
\end{equation*}
\end{definition}

We are now ready to finish the proof. First, we relate $\norm{A}_2$ to $\val(f_{L,R})$.
Fix an assignment $x \in \Fits^n$, and let $z \in \Fits^N$ be defined as $z_{S} \coloneqq \prod_{u \in S_1} x_u \prod_{v \in S_2} x_v$ for $S = S_1^{(1)} \cup S_2^{(2)} \subseteq [n] \times [2]$ satisfying $\abs{S} = \ell$. 

We observe that $D \cdot f_{L,R}(x) = z^{\top} A z$. This is because:
\begin{enumerate}[(1)]
    \item For $S, T \subseteq [n] \times [2]$ with  $S \oplus T = C^{(1)} \oplus C'^{(2)}$, we have \\ $z_S z_T = \prod_{u \in S_1} x_u \prod_{v \in S_2} x_v \prod_{u' \in T_1} x_u \prod_{v' \in T_2} x_v = \prod_{u \in S_1 \oplus T_1} x_u \prod_{v \in S_2 \oplus T_2} x_v = \prod_{u \in C} x_u \prod_{v \in C'} x_v$,
    \item For a pair of clauses $(u,C) \in \cH_i$ and $(u, C') \in \cH_j$ with $i \in L$ and $j \in R$, there are exactly $D = 2{2n - 4 \choose \ell - 2}$ nonzero entries $(S,T)$ of $A^{(i, C, C')}$, and these entries have $S \oplus T = C^{(1)} \oplus C'^{(2)}$.
\end{enumerate}
In particular, this implies
\begin{equation}\label{eq:boolnorm} 
\val(f_{L,R}) \leq \frac{N}{D}\cdot  \norm{A}_2\enspace.
 \end{equation}
It thus remains to bound $\E_{b \in \Fits^k}[\norm{A}_2]$, which we do in the following lemma.
\begin{lemma}[Spectral norm bound]
\label{lem:specbound}
$\E_{b \in \Fits^k} [\norm{A}_2] \leq  d \cdot O(\sqrt{k \ell \log n})$.
\end{lemma}
We postpone the proof of \cref{lem:specbound} to \cref{sec:specbound}, and now finish the proof of \cref{lem:3xor}.
\begin{proof}[Proof of~\cref{lem:3xor}]
By \cref{eq:boolnorm,lem:specbound}, we have that
\begin{flalign*}
&\E_{b \in \Fits^k} [\val(f_{L,R})] \leq \frac{N}{D} \E_{b \in \Fits^k}[\norm{A}_2]  \\
&\leq \frac{N}{D} \left( d \cdot O(\sqrt{k \ell \log n})\right) \leq \frac{n^2}{\ell^2} d \cdot O(\sqrt{k \ell \log n}) \\
&= nkd \cdot O((nk)^{1/4} \sqrt{\log n}) \enspace,
\end{flalign*}
where we use that $\ell = (\sqrt{n/k})/c$ for some constant $c$, and we use \cref{fact:binomialratio} to bound $N/D$.
Finally, combining with \cref{lem:cauchyschwarz} and using that $m \leq n k$, we have that
\begin{flalign*}
\E[\val(f)]^2 \leq \E[\val(f)^2] &\leq \frac{1}{9} \cdot \left(3n^2 k + 4n\E_{(L,R)}\E_{b \in \Fits^k} [\val(f_{L,R})]\right) \\
&\leq n^2 k d \cdot O((nk)^{1/4} \sqrt{\log n})  \enspace .
\end{flalign*}
Hence, 
\begin{flalign*}
\E[\val(f)] \leq n \sqrt{kd} \cdot O\left((nk)^{1/8} \log^{1/4} n\right) \enspace,
\end{flalign*}
which finishes the proof of \cref{lem:3xor}.
\end{proof}

\subsection{Counting nonzero entries: proof of \cref{lem:rowpruning}}
\label{sec:rowpruning}
\begin{proof}[Proof of \cref{lem:rowpruning}]
Fix $j \in R$ and clauses $(u,C) \in \cH_i$ and $(u, C') \in \cH_j$. Recall that in \cref{rem:edge-count}, we observed that there are exactly $4 {2n - 4 \choose \ell - 2}$ pairs $(S,T)$ with $S \overset{C,C'}{\leftrightarrow} T$. Indeed, this is because $S \overset{C,C'}{\leftrightarrow} T$ if and only if $S$ and $T$ each contain one entry of $C$ and $C'$ ($2$ choices per clause), and the remaining part of $S$ and $T$ is the same set $Q \subseteq [n] \times [2] \setminus (C^{(1)} \oplus C'^{(2)})$ of size $\ell - 2$ (which has ${2n - 4 \choose \ell - 2}$ choices).

From the above, we observe that for each $Q \subseteq [n] \times [2] \setminus (C^{(1)} \oplus C'^{(2)})$ of size $\ell - 2$, we can identify $Q$ with $4$ different pairs $(S,T)$ with $S \overset{C,C'}{\leftrightarrow} T$; namely, each pair $(S,T)$ corresponds to a subset of size $2$ of $(C,C')$ containing exactly one entry from each of $C, C'$. We note that these $4$ choices of $(S,T)$ correspond exactly to the $4$ half clauses in $P_i$ contributed by the derived clause $(C,C')$. We will show that for at least $\frac{1}{2} {2n - 4 \choose \ell - 2}$ choices of $Q$, \emph{all} $4$ corresponding choices of $(S,T)$ will contain exactly one derived clause from $P_i$: namely, the half clause of $(C,C')$ that we add to $Q$ to obtain $S$ or $T$. This clearly suffices to finish the proof.

Call such a set $Q$ \emph{bad} if it does not have the above property, i.e., there is some pair $(S,T)$ identified with $Q$ such that one of $S$ or $T$ contains more than one half clause from $P_i$. Since $S \overset{C,C'}{\leftrightarrow} T$ already implies that each of $S$ and $T$ has exactly one half clause from $C^{(1)} \oplus C'^{(2)}$, there are three ways that $Q$ can be bad:
\begin{enumerate}[(1)]
\item $Q$ contains a half clause from $P_i$,
\item there is $v^{(1)} \in C^{(1)}$ and $w^{(2)} \in Q$ such that $(v^{(1)}, w^{(2)}) \in P_i$,
\item there is $v^{(1)} \in Q$ and $w^{(2)} \in C'^{(2)}$ such that $(v^{(1)}, w^{(2)}) \in P_i$.
\end{enumerate}
We thus have that the number of bad $Q$'s is at most 
\begin{equation*}
p_0 {2n - 6 \choose \ell - 4} + p_1 {2n - 5 \choose \ell - 3} + p_2 {2n - 5 \choose \ell - 3} \enspace,
\end{equation*}
where $p_{0} = \abs{P_i}$, $p_{1} = \abs{\{(v^{(1)}, w^{(2)}) \in P_i : v^{(1)} \in C^{(1)}\}}$, $p_{2} = \abs{\{(v^{(1)}, w^{(2)}) \in P_i : w^{(2)} \in C'^{(2)}\}}$. 

We now upper bound $p_0, p_1, p_2$. Recall that a half clause in $P_i$ is a pair $(v^{(1)}, w^{(2)})$ such that there are clauses $(u,C_1) \in \cH_i$, $(u, C_2) \in \cH_j$ with $j \in R$, and $v \in C_1$, $w \in C_2$.
\begin{enumerate}[(1)]
\item We have $p_0 \leq 4nk$, as for each $u \in [n]$, because the $\cH_i$'s are matchings, there is at most one $C_1$ such that $(u, C_1) \in \cH_i$, and at most $k$ choices of $(u, C_2) \in \cH_j$ with $j \in R$, as $\abs{R} \leq k$. Finally, each choice of $(C_1, C_2)$ yields $4$ half clauses.
\item We have $p_1 \leq 8k$. First, there are at most $2$ choices for $v$, each coming from $C$. For each such $v$, there is at most one $C_i \in \cH_i$ with $v \in C_i$. (Note that $\abs{C_i} = 3$.) Once $C_i$ is fixed, we have at most $2$ choices for $u$, given by $C_i \setminus \{v\}$, and there are at most $k$ hyperedges $(u, C_2) \in \cH_j$ for $j \in R$ (as each $\cH_j$ is a matching and $\abs{R} \le k$). Finally, for each such $C_2$ there are $2$ possible choices for $w$.
\item We have $p_2 \leq 8k$. First, there are at most $2$ choices for $w$, each coming from $C'$.  For each such $w$, there are at most $k$ choices of $C_j \in \cup_{j \in R} \cH_j$ with $w \in C_j$, as each $\cH_j$ is a matching and $\abs{R} \le k$. (Note that $\abs{C_j} = 3$.) For each such $C_j$, there are at most $2$ choices for $u$, given by $C_j \setminus \{w\}$, and for each $u$, there is at most one choice of $C_1$ such that $(u, C_1) \in \cH_i$. Finally, such a $C_1$, if it exists, gives $2$ choices for $v$. 
\end{enumerate}
Combining, we thus have that the number of bad $Q$'s is at most
\begin{equation*}
4nk {2n - 6 \choose \ell - 4} + 16k {2n - 5 \choose \ell - 3} \enspace.
\end{equation*}
We have that
\begin{flalign*}
&\frac{4nk {2n - 6 \choose \ell - 4} + 16k {2n - 5 \choose \ell - 3}}{{{{2n - 4} \choose \ell - 2}}} = \frac{4nk \frac{(2n - 6)!}{(\ell - 4)! (2n - 2 - \ell)!} + 16k \frac{(2n - 5)!}{(\ell - 3)!(2n - 2 - \ell)!}}{\frac{(2n - 4)!}{(\ell - 2)!(2n - 2 - \ell)!}} \\
&= 4nk \frac{(\ell - 2)(\ell - 3)}{(2n - 4)(2n - 5)} + 16 k \frac{\ell - 2}{2n - 4} \leq \frac{1}{2} \enspace,
\end{flalign*}
as we have $\ell \leq (\sqrt{n/k}) / c$, for some sufficiently large constant $c$, and $k \leq \sqrt{nk}$ since $k \leq n$.
\end{proof}

\subsection{Spectral norm bound: proof of \cref{lem:degbound,lem:specbound}}
\label{sec:specbound}
\begin{proof}[Proof of \cref{lem:degbound}]
Fix $i \in L$. We show that each row/column of $B_i$ has at most $2d$ nonzero entries. Indeed, this is because if $S$ is a nonzero row (or column) in $B_i$, then $S$ contains at most one half clause from $P_i$. If $(C,C')$ is a derived clause where $S \overset{C,C'}{\leftrightarrow} T$ for some $T$, then $S$ must contain a half clause in $P_i$ that is contained in $C^{(1)} \oplus C'^{(2)}$, i.e., a half clause coming from $(C,C')$. As $S$ contains at most one half clause, it follows that the number of nonzero entries in the $S$-th row is upper bounded by the maximum, over all half clauses, of the number of derived clauses $(C,C')$ that contain this half clause. One can observe that this is $2d$. Indeed, if we fix $v^{(1)}$ and $w^{(2)}$, there is at most one clause $C \in \cH_i$ containing $v$. Once $v$ is fixed, there are two choices for $u$ in $C \setminus \{v\}$. Once we have chosen $u$, the second clause must be $(u, C') \in \cH_j$ for some $j \in R$, where $C'$ contains $w$. By assumption, the number of hyperedges in $\cup_{i = 1}^k \cH_i$ containing the pair $\{u,w\}$ is at most $d$, so there are at most $d$ choices for $C'$. 
\end{proof}

\begin{proof}[Proof of \cref{lem:specbound}]
We have that $A = \sum_{i \in L} b_i A_i$, where the $b_i$'s are i.i.d.\ from $\Fits$. By \cref{lem:degbound}, we know that the number of nonzero entries in a row/column of $B_i$ is at most $2d$. As $A_i$ is obtained by zero-ing out entries of $B_i$, it follows that this also holds for $A_i$. It thus follows that the $\ell_1$-norm of any row/column of $A_i$ is at most $2d$, and thus $\norm{A_i}_2 \leq 2d$. This additionally implies that $\norm{\sum_{i \in L} A_i A_i^{\top}}_2 \leq \abs{L} (2d)^2 \leq k (2d)^2$, and that $\norm{\sum_{i \in L} A_i^{\top} A_i}_2 \leq \abs{L} (2d)^2 \leq k (2d)^2$. Applying Matrix Khintchine (\cref{fact:matrinxkhintchine}), we conclude that $\E[\norm{A}_2] \leq d \cdot O(\sqrt{k \log N})$. As  $\log N = O(\ell \log n)$, \cref{lem:specbound} follows.
\end{proof}


\section{CSP Refutation Proof of Existing LDC Lower Bounds}
\label{sec:even-q}
In this section, we prove the following theorem, which are the existing LDC lower bounds using the connection between LDCs and CSP refutation.
\begin{theorem}
\label{thm:knownlowerbound}
Let $\cC \colon \Bits^k \to \Bits^n$ be a code that is $(q, \delta, \eps)$-locally decodable, for constant $q \geq 2$. Then, the following hold: \begin{enumerate}[(1)]
\item If $q$ is even, $k \leq n^{1 - 2/q} O((\log n)/\eps^4 \delta^2)$, and 
\item If $q$ is odd, $k \leq n^{1 - 2/(q+1)} O((\log n)/\eps^4 \delta^2)$.
\end{enumerate}
\end{theorem}
\begin{proof}
By \cref{fact:normalform}, it suffices to show that for a code $\cC \colon \Fits^k \to \Fits^n$ that is $(q, \delta, \eps)$-normally decodable, it holds that \begin{inparaenum}[(1)]
\item $k \leq n^{1 - 2/q} O((\log n)/\eps^2 \delta^2)$ if $q$ is even, and 
\item $k \leq n^{1 - 2/(q+1)} O((\log n)/\eps^2 \delta^2)$ if $q$ is odd.\end{inparaenum}

We first observe for any $q$, we can transform $\cC$ into a code $\cC'$ that is $(q+1, \delta/2, \eps)$-normally decodable. In particular, it suffices to prove the lower bound in the case when $q$ is even. We note that one can also prove the $q$ odd case directly using a similar approach to the even case, just with asymmetric matrices. For simplicity, we do not present this proof, but the definition of the asymmetric matrices is given in \cref{remark:oddmatrix}.
\begin{claim}
Let $\cC \colon \Fits^k \to \Fits^n$ be a code that is $(q, \delta, \eps)$-normally decodable. Then, there is a code $\cC' \colon \Fits^k \to \Fits^{2n}$ that is $(q+1, \delta/2, \eps)$-normally decodable.
\end{claim}
\begin{proof}
Let $\cC' \colon \Fits^k \to \Fits^{2n}$ be defined by setting $\cC'(b) = \cC(b) \| 1^n$, i.e., the encoding of $b$ under the original code $\cC$ concatenated with $n$ $1$'s. For each hypergraph $\cH_i$, we construct the hypergraph $\cH'_i$ as follows. First, let $\pi_i \colon \cH_i \to [n]$ be an arbitrary ordering of the hyperedges of $\cH_i$, and then let $\cH'_i = \{C \cup \{n + \pi_i(C)\} : C\in \cH_i\}$. That is, the hypergraph $\cH'_i$ is obtained by taking each hyperedge in $\cH_i$ and appending one of the new coordinates, and each new coordinate is added to at most one hyperedge, so that $\cH'_i$ remains a matching. It is now obvious from construction that $\cC'$ is $(q+1, \delta/2, \eps)$-normally decodable, which finishes the proof.
\end{proof}

It thus remains to show that for any code $\cC \colon \Fits^k \to \Fits^n$ that is $(q, \delta, \eps)$-normally decodable with $q$ even, it holds that $n \geq \tilde{\Omega}(k^{\frac{q}{q-2}})$ for $q \ge 4$ and $n \ge \exp(\Omega(k))$ for $q = 2$. Without loss of generality, we may assume that the hypergraphs $\cH_1, \dots, \cH_k$ all have size \emph{exactly} $\delta n$.

Similar to the proof of \cref{mthm:main}, we construct a $q$-XOR instance associated with $\cC'$, and argue via CSP refutation that its value must be small. For each $b \in \Fits^k$, let $\Psi_b$ denote the $q$-XOR instance with variables $x \in \Fits^n$ and constraints $\prod_{v \in C} x_v = b_i$ for all $i \in [k], C \in \cH_i$. We let $m \coloneqq \sum_{i = 1}^k \abs{\cH_i}$ denote the total number of constraints. Let $\psi_b(x) \coloneqq \frac{1}{m} \sum_{i = 1}^k b_i \sum_{C \in \cH_i} \prod_{v \in C} x_v$, and let $\val(\psi_b) \coloneqq \max_{x \in \Fits^n} \psi_b(x)$. As in the proof of \cref{mthm:main}, we observe that \cref{def:normalLDC} implies that $\E_{b \gets \Fits^k}[\val(\psi_b)] \geq 2\eps$.

It thus remains to upper bound $\E_{b \gets \Fits^k}[\val(\psi_b)]$. We do this by introducing a matrix $A$ for each $b \in \Fits^k$, where $\norm{A}_2$ is related to $\val(\psi_b)$. We then upper bound $\E_{b\gets\Fits^k}[\norm{A}_2]$. We note that the matrix $A$ depends on the choice of $b \in \Fits^k$ but we suppress this dependence for notational simplicity. 

\begin{definition}
\label{def:kikuchiqeven}
Let $\ell \coloneqq n^{1 - 2/q}/c$ for some absolute constant $c \geq e^{16}$, and let $N \coloneqq {n \choose \ell}$. For each $q$-uniform hypergraph matching $\cH_i$, let $A_i \in \R^{N \times N}$ denote the matrix indexed by sets $S, T \in {[n] \choose \ell}$ where $A_i(S,T) = 1$ if the pair $(S,T)$ satisfies \begin{inparaenum}[(1)] \item $S \oplus T = C \in \cH_i$, and \item $\abs{S \oplus C'} \ne \ell$, $\abs{T \oplus C'} \ne \ell$ for every $C' \in \cH_i$ with $C' \ne C$\end{inparaenum}. We set $A_i(S,T) = 0$ otherwise. We let $A \coloneqq \sum_{i = 1}^k b_i A_i$.
\end{definition}
\begin{remark}[Matrices for $q$ odd]
\label{remark:oddmatrix}
As mentioned earlier, when $q$ is odd we can prove the lower bound directly by choosing slightly different matrices, although we do not present the proof in full. The matrices used are defined as follows. We let the matrix $A_i$ now be indexed by rows $S \in {[n] \choose \ell}$ and columns $T \in {[n] \choose \ell + 1}$, and let $A_i(S,T) = 1$ if $S \oplus T = C \in \cH_i$, and $\abs{S \oplus C'} \ne \ell + 1, \abs{T \oplus C'} \ne \ell$, for all $C' \in \cH_i$ with $C' \ne C$. The matrix $A$ is again defined as $\sum_{i = 1}^k b_i A_i$.
\end{remark}

\begin{lemma}
\label{lem:qevenrowpruning}
There is an integer $D$ such that the following holds.
Fix $i \in [k]$, and let $A_i$ be one of the matrices defined in \cref{def:kikuchiqeven}. For any $C \in \cH_i$, the number of pairs $(S,T)$ with $S \oplus T = C$ and $A_i(S,T) = 1$ is exactly $D$. Moreover, we have that $D/N \geq  \frac{1}{2} {q \choose q/2} e^{-3q}   \cdot (\frac{\ell}{n})^{q/2}$.
\end{lemma}
We postpone the proof of \cref{lem:qevenrowpruning}, and now finish the proof.

Our proof now proceeds as in \cref{sec:3xor}. We similarly observe that $\val(\psi_b) \leq \frac{N}{m D} \norm{A}_2$, where $D$ is from \cref{lem:qevenrowpruning}, and $m \coloneqq \sum_{i = 1}^k \abs{\cH_i}$ is the total number of constraints. It thus remains to bound $\E_{b \gets \Fits^k}[\norm{A}_2]$, which we do in the following lemma.
\begin{lemma}[Spectral norm bound]
\label{lem:qevenspecbound}
$\E_{b \in \Fits^k} [\norm{A}_2] \leq O(\sqrt{k \ell \log n})$.
\end{lemma}
\begin{proof}
We will use Matrix Khintchine (\cref{fact:matrinxkhintchine}) to bound $\E[\norm{A}_2]$. We have $A = \sum_{i = 1}^k b_i A_i$. We observe that $\norm{A_i}_2 \leq 1$ by construction, as the $\ell_1$-norm of any row/column of $A_i$ is at most $1$. It then follows that $\norm{\sum_{i = 1}^k A_i^2}_2 \leq \sum_{i = 1}^k \norm{A_i}_2^2 \leq k$. Hence, by \cref{fact:matrinxkhintchine}, it follows that $\E[\norm{A}_2] \leq O(\sqrt{k \log N})$. Finally, we observe that $\log_2 N \leq \ell \log_2 n$, which finishes the proof.
\end{proof}
We now finish the proof of \cref{thm:knownlowerbound}.
By \cref{lem:qevenspecbound}, we have
\begin{flalign*}
2\eps \leq \E_{b \in \Fits^k}[\val(\psi_b)] \leq \frac{1}{mD} N O(\sqrt{k \ell \log n}) \enspace.
\end{flalign*}
As $\abs{\cH_i} = \delta n$ for all $i$, it follows that $m = \delta nk$. Therefore,
\begin{flalign*}
\eps \leq \frac{N}{\delta nk D}  O(\sqrt{k \ell \log n}) &\leq \frac{1}{\delta n k} \left(\frac{n}{\ell}\right)^{q/2} \cdot O(\sqrt{k \ell \log n}) \leq \frac{1}{\delta}  \cdot O\left(\sqrt{\frac{ n^{1-2/q}}{k} \log n}\right) \enspace,
\end{flalign*}
where we use that $\ell = n^{1-2/q}/c$ and the bound on $\frac{D}{N}$ from \cref{lem:qevenrowpruning}.
We thus conclude that $k \leq n^{1 - 2/q} \cdot O(\log n)/\eps^2 \delta^2$.
\end{proof}

\begin{proof}[Proof of \cref{lem:qevenrowpruning}]
First, let $C \in \cH_i$ be any element. We first show that the number of pairs $(S,T)$ with $S \oplus T = C$ and $A_i(S,T) = 1$ is independent of $C$. Indeed, let $C' \in \cH_i$ be different from $C$. As $\cH_i$ is a matching, we have that $C$ and $C'$ are disjoint. Let $\pi$ be an arbitrary bijection between $C$ and $C'$ and extend $\pi$ to act on all of $[n]$ by acting as the identity on elements not in $C \cup C'$. It is simple to observe that if $(S,T)$ is any pair satisfying the above criterion for $C$, then $(S',T')$, obtained by applying $\pi$ to all elements of $S$ and $T$, satisfies the criterion for $C'$. Hence, the number of pairs is independent of the choice of $C \in \cH_i$.

We note that it is clear from symmetry that $D$ depends only on $\abs{\cH_i}$, $q$, and $n$. As $\abs{\cH_i} = \delta n$ for all $i$, it follows that $D$ does not depend on $i$.

We now lower bound $D$. Let $C \in \cH_i$ be arbitrary. We observe that $S \oplus T = C$ if and only if $S = C_S \cup Q$, $T = C_T \cup Q$, where $C_S, C_T \subseteq C$ are disjoint subsets of size exactly $q/2$,
so that $C = C_S \cup C_T$,
$Q \subseteq [n] \setminus C$
has size exactly $\ell - q/2$. It follows that if $S \oplus T = C$ and for some $C' \ne C \in \cH_i$, either $\abs{S \oplus C'} = \ell$ or $\abs{T \oplus C'} = \ell$, then it must be the case that $\abs{Q \cap C'} = q/2$. Hence, we have that
\begin{equation*}
D \geq {q \choose q/2}{{n - q}\choose{\ell - q/2}} - \abs{\cH_i} \cdot {q \choose q/2}^2 {{n - 2q}\choose{\ell - q}} \enspace.
\end{equation*}
Applying \cref{fact:binomialratio}, we thus have that
\begin{align*}
D/N & \geq {q \choose q/2} e^{-3q} \left(\frac{\ell}{n}\right)^{q/2} 
- n \cdot {q \choose q/2}^2 e^{3q} \left(\frac{\ell}{n}\right)^{q} \\ 
& = {q \choose q/2} e^{-3q} \left(\frac{\ell}{n}\right)^{q/2} \left(1 -  n \cdot 2^{q} e^{6q} \left(\frac{\ell}{n}\right)^{q/2} \right)\\
& \geq \frac{1}{2} {q \choose q/2} e^{-3q} \left(\frac{\ell}{n}\right)^{q/2} \enspace,
\end{align*}
where we use that $\ell \leq n^{1-2/q}/e^{16}$.
\end{proof}

\section*{Acknowledgements}
We thank the anonymous reviewers for their helpful comments
on an earlier draft of the paper. We also thank Tim Hsieh and Sidhanth Mohanty for helpful discussions.

\bibliographystyle{alpha}
\bibliography{bib/custom,bib/dblp,bib/mathreview,bib/scholar,bib/references,bib/witmer, ldc-lb.bib}

\newcommand{\etalchar}[1]{$^{#1}$}
\begin{thebibliography}{ALM{\etalchar{+}}98}

\bibitem[AGK21]{AbascalGK21}
Jackson Abascal, Venkatesan Guruswami, and Pravesh~K. Kothari.
\newblock Strongly refuting all semi-random boolean csps.
\newblock In {\em Proceedings of the 2021 {ACM-SIAM} Symposium on Discrete
  Algorithms, {SODA} 2021, Virtual Conference, January 10 - 13, 2021}, pages
  454--472. {SIAM}, 2021.

\bibitem[ALM{\etalchar{+}}98]{ALMSS98}
Sanjeev Arora, Carsten Lund, Rajeev Motwani, Madhu Sudan, and Mario Szegedy.
\newblock Proof verification and the hardness of approximation problems.
\newblock {\em Journal of the ACM (JACM)}, 45(3):501--555, 1998.

\bibitem[AS98]{AS98}
Sanjeev Arora and Shmuel Safra.
\newblock Probabilistic checking of proofs: A new characterization of np.
\newblock {\em Journal of the ACM (JACM)}, 45(1):70--122, 1998.

\bibitem[BCG20]{BCG20}
Arnab Bhattacharyya, L~Sunil Chandran, and Suprovat Ghoshal.
\newblock Combinatorial lower bounds for 3-query ldcs.
\newblock In {\em 11th Innovations in Theoretical Computer Science Conference
  (ITCS 2020)}, volume 151, page~85. Schloss Dagstuhl--Leibniz-Zentrum fuer
  Informatik, 2020.

\bibitem[BGT17]{BGT17}
Arnab Bhattacharyya, Sivakanth Gopi, and Avishay Tal.
\newblock Lower bounds for 2-query lccs over large alphabet.
\newblock In {\em Approximation, Randomization, and Combinatorial Optimization.
  Algorithms and Techniques (APPROX/RANDOM 2017)}. Schloss
  Dagstuhl-Leibniz-Zentrum fuer Informatik, 2017.

\bibitem[Bri16]{Bri16}
Jop Bri{\"e}t.
\newblock On embeddings of $\ell_1^k$ from locally decodable codes.
\newblock {\em arXiv preprint arXiv:1611.06385}, 2016.

\bibitem[CGW10]{ChenGW10}
Victor Chen, Elena Grigorescu, and {Ronald de} Wolf.
\newblock Efficient and error-correcting data structures for membership and
  polynomial evaluation.
\newblock In {\em 27th International Symposium on Theoretical Aspects of
  Computer Science, {STACS} 2010, March 4-6, 2010, Nancy, France}, volume~5 of
  {\em LIPIcs}, pages 203--214. Schloss Dagstuhl - Leibniz-Zentrum f{\"{u}}r
  Informatik, 2010.

\bibitem[DGY11]{DGY11}
Zeev Dvir, Parikshit Gopalan, and Sergey Yekhanin.
\newblock Matching vector codes.
\newblock {\em SIAM Journal on Computing}, 40(4):1154--1178, 2011.

\bibitem[DS05]{DvirS05}
Zeev Dvir and Amir Shpilka.
\newblock Locally decodable codes with 2 queries and polynomial identity
  testing for depth 3 circuits.
\newblock In {\em Proceedings of the 37th Annual {ACM} Symposium on Theory of
  Computing, Baltimore, MD, USA, May 22-24, 2005}, pages 592--601. {ACM}, 2005.

\bibitem[Dvi10]{Dvir10}
Zeev Dvir.
\newblock On matrix rigidity and locally self-correctable codes.
\newblock In {\em Proceedings of the 25th Annual {IEEE} Conference on
  Computational Complexity, {CCC} 2010, Cambridge, Massachusetts, USA, June
  9-12, 2010}, pages 291--298. {IEEE} Computer Society, 2010.

\bibitem[Efr09]{Efremenko09}
Klim Efremenko.
\newblock 3-query locally decodable codes of subexponential length.
\newblock In {\em Proceedings of the 41st Annual {ACM} Symposium on Theory of
  Computing, {STOC} 2009, Bethesda, MD, USA, May 31 - June 2, 2009}, pages
  39--44. {ACM}, 2009.

\bibitem[Fei08]{Fei08}
Uriel Feige.
\newblock Small linear dependencies for binary vectors of low weight.
\newblock In {\em Building bridges}, volume~19 of {\em Bolyai Soc. Math.
  Stud.}, pages 283--307. Springer, Berlin, 2008.

\bibitem[GKM22]{GuruswamiKM22}
Venkatesan Guruswami, Pravesh~K. Kothari, and Peter Manohar.
\newblock Algorithms and certificates for boolean {CSP} refutation: smoothed is
  no harder than random.
\newblock In {\em {STOC} '22: 54th Annual {ACM} {SIGACT} Symposium on Theory of
  Computing, Rome, Italy, June 20 - 24, 2022}, pages 678--689. {ACM}, 2022.

\bibitem[GKST06]{GKST06}
Oded Goldreich, Howard Karloff, Leonard~J Schulman, and Luca Trevisan.
\newblock Lower bounds for linear locally decodable codes and private
  information retrieval.
\newblock {\em Computational Complexity}, 15(3):263--296, 2006.

\bibitem[Gop18]{Gop18}
Sivakanth Gopi.
\newblock {\em Locality in Coding Theory}.
\newblock PhD thesis, Princeton University, 2018.

\bibitem[Gop19]{Gopi19}
Sivakanth Gopi.
\newblock Modern coding theory: lecture notes and exercises, 2019.
\newblock URL:
  \url{https://homes.cs.washington.edu/~anuprao/pubs/codingtheory/exercise2.pdf}.

\bibitem[HKM23]{HsiehKM23}
Jun{-}Ting Hsieh, Pravesh~K. Kothari, and Sidhanth Mohanty.
\newblock A simple and sharper proof of the hypergraph moore bound.
\newblock {\em ACM-SIAM Symposium on Discrete Algorithms, SODA}, 2023.

\bibitem[IK04]{IshaiK04}
Yuval Ishai and Eyal Kushilevitz.
\newblock On the hardness of information-theoretic multiparty computation.
\newblock In {\em Advances in Cryptology - {EUROCRYPT} 2004, International
  Conference on the Theory and Applications of Cryptographic Techniques,
  Interlaken, Switzerland, May 2-6, 2004, Proceedings}, volume 3027 of {\em
  Lecture Notes in Computer Science}, pages 439--455. Springer, 2004.

\bibitem[KT00]{KT00}
Jonathan Katz and Luca Trevisan.
\newblock On the efficiency of local decoding procedures for error-correcting
  codes.
\newblock In {\em Proceedings of the thirty-second annual ACM symposium on
  Theory of computing}, pages 80--86, 2000.

\bibitem[KW04]{KdW04}
Iordanis Kerenidis and {Ronald de} Wolf.
\newblock Exponential lower bound for 2-query locally decodable codes via a
  quantum argument.
\newblock {\em Journal of Computer and System Sciences}, 69(3):395--420, 2004.

\bibitem[O'D14]{OD14}
Ryan O'Donnell.
\newblock {\em Analysis of Boolean Functions}.
\newblock Cambridge University Press, 2014.

\bibitem[Rom06]{Romashchenko06}
Andrei~E. Romashchenko.
\newblock Reliable computations based on locally decodable codes.
\newblock In {\em {STACS} 2006, 23rd Annual Symposium on Theoretical Aspects of
  Computer Science, Marseille, France, February 23-25, 2006, Proceedings},
  volume 3884 of {\em Lecture Notes in Computer Science}, pages 537--548.
  Springer, 2006.

\bibitem[SS12]{schudysviridenko}
Warren Schudy and Maxim Sviridenko.
\newblock Concentration and moment inequalities for polynomials of independent
  random variables.
\newblock In {\em Proceedings of the Twenty-Third Annual ACM-SIAM Symposium on
  Discrete Algorithms}, SODA '12, page 437–446, USA, 2012. Society for
  Industrial and Applied Mathematics.

\bibitem[Tre04]{Tre04}
Luca Trevisan.
\newblock Some applications of coding theory in computational complexity.
\newblock {\em arXiv preprint cs/0409044}, 2004.

\bibitem[Tro15]{Tropp15}
Joel~A. Tropp.
\newblock An introduction to matrix concentration inequalities.
\newblock {\em Found. Trends Mach. Learn.}, 8(1-2):1--230, 2015.

\bibitem[WAM19]{WeinAM19}
Alexander~S. Wein, Ahmed~El Alaoui, and Cristopher Moore.
\newblock The kikuchi hierarchy and tensor {PCA}.
\newblock In {\em 60th {IEEE} Annual Symposium on Foundations of Computer
  Science, {FOCS} 2019, Baltimore, Maryland, USA, November 9-12, 2019}, pages
  1446--1468. {IEEE} Computer Society, 2019.

\bibitem[Wol09]{Wolf09}
{Ronald de} Wolf.
\newblock Error-correcting data structures.
\newblock In {\em 26th International Symposium on Theoretical Aspects of
  Computer Science, {STACS} 2009, February 26-28, 2009, Freiburg, Germany,
  Proceedings}, volume~3 of {\em LIPIcs}, pages 313--324. Schloss Dagstuhl -
  Leibniz-Zentrum f{\"{u}}r Informatik, Germany, 2009.

\bibitem[Woo07]{Woo07}
David Woodruff.
\newblock New lower bounds for general locally decodable codes.
\newblock In {\em Electronic Colloquium on Computational Complexity (ECCC)},
  volume~14, 2007.

\bibitem[Woo12]{Woo12}
David~P Woodruff.
\newblock A quadratic lower bound for three-query linear locally decodable
  codes over any field.
\newblock {\em Journal of Computer Science and Technology}, 27(4):678--686,
  2012.

\bibitem[Yek08]{Yek08}
Sergey Yekhanin.
\newblock Towards 3-query locally decodable codes of subexponential length.
\newblock {\em Journal of the ACM (JACM)}, 55(1):1--16, 2008.

\bibitem[Yek10]{Yekhanin10}
Sergey Yekhanin.
\newblock {\em Locally Decodable Codes and Private Information Retrieval
  Schemes}.
\newblock Information Security and Cryptography. Springer, 2010.

\bibitem[Yek12]{Yek12}
Sergey Yekhanin.
\newblock Locally decodable codes.
\newblock {\em Foundations and Trends in Theoretical Computer Science},
  6(3):139--255, 2012.

\end{thebibliography}


\appendix


\section{Improved Lower Bounds for $3$-LDCs over Larger Alphabets}
\label{sec:general-alphabets}

In this appendix, we will extend \cref{mthm:main} to $3$-query LDCs over larger alphabets, which will follow from combining \cref{mthm:main} with standard results from~\cite{KT00, KdW04}. We first define LDCs over general alphabets.

\begin{definition}[LDCs over general alphabets]
\label{def:general-ldc}
Given a positive integer $q$, constants $\delta, \eps > 0$, and an alphabet $\Sigma$, we say a code $\cC \colon \Bits^k \to \Sigma^n$ is $(q,\delta,\eps)$-locally decodable code (abbreviated $(q,\delta,\eps)$-LDC) if there exists a randomized decoding algorithm $\Dec(\cdot)$ with the following properties. The algorithm $\Dec(\cdot)$ is given oracle access to some $y \in \Sigma^n$, takes an $i \in [k]$ as input, and satisfies the following: \begin{inparaenum}[(1)] \item the algorithm $\Dec$ makes at most $q$ queries to the string $y$, and \item for all $b \in \Bits^k$, $i \in [k]$, and all $y \in \Sigma^n$ such that $\Delta(y, \cC(b)) \leq \delta n$, $\Pr[\Dec^{y}(i) = b_i] \geq \frac{1}{2} + \eps$. \end{inparaenum}
\end{definition} 

Our extension of \cref{mthm:main} to larger alphabets is the following theorem.

\begin{theorem}
\label{thm:main-gen-alpha}
Let $\cC\colon\Bits^k \to \Sigma^n$ be a $(3, \delta, \eps)$-LDC. Then, it must hold that $k^3 \leq \abs{\Sigma}^{41}n \cdot O(\log^6(\abs{\Sigma} n)/\eps^{32} \delta^{16})$. In particular, if $\delta, \eps$ are constants and $\abs{\Sigma} \leq n$, then $n \geq \Omega(k^3/(\abs{\Sigma}^{41}\log^6{k}))$.
\end{theorem}

To prove \cref{thm:main-gen-alpha}, it suffices to show the following lemma.

\begin{lemma}
\label{lemma:gen-normal-form}
Let $\cC\colon\Bits^k \to \Sigma^n$ be a $(3, \delta, \eps)$-LDC. Then, there exists a binary code $\cC'\colon\Bits^k \to \Bits^{n'}$ with $n' \le 4n\abs{\Sigma}$ and $q$-uniform matchings $\cH_1', \ldots , \cH_k'$ over $n'$ vertices such that for all $i \in [k]$, we have $\abs{\cH_i'} \ge \eps \delta n'/(4q^2\abs{\Sigma})$. Furthermore, for any query set $C \in \cH_i'$, we have that $\Pr_{b \gets \Bits^k}[b_i = \oplus_{v \in C}{\cC(b)_v}] \geq \frac{1}{2} + \frac{\eps}{8\abs{\Sigma}^{3/2}}$.
\end{lemma}

Indeed, once we have \cref{lemma:gen-normal-form}, then by applying \cref{mthm:main} on the resulting normal LDC,\footnote{Note that we obtain a better dependence on $\eps$ in \cref{mthm:main} when our initial LDC is in normal form, as shown at the beginning of \cref{sec:proof}.} we obtain \cref{thm:main-gen-alpha}. Now, to prove \cref{lemma:gen-normal-form}, we first need the following result from~\cite{KT00}.

\begin{lemma}[Theorem 1 + Lemma 4 in \cite{KT00}]
\label{lemma:kt}
Let $\cC\colon\Bits^k \to \Sigma^n$ be a $(q, \delta, \eps)$-LDC. Then, there exists $q$-uniform matchings $\cH_1, \ldots , \cH_k$ over $[n]$ such that for all $i \in [k]$, we have $\abs{\cH_i} \ge \eps \delta n/q^2$. Furthermore, for any query set $C \in \cH_i$, there exists a function $f_C\colon\Sigma^q \to \Bits$ such that $\Pr_{b \gets \Bits^k}[b_i = f_C(\cC(b)|_C)] \geq \frac{1}{2} + \frac{\eps}{2}$.
\end{lemma}
Note that formally the statement in~\cite{KT00} only guarantees that each query set in $\cH_i$ has size at most $q$ rather than \emph{exactly} $q$. However, we can trivially make each set be of size exactly $q$ by padding each codeword of $\cC$ with $n$ zeros.

Next, we need the following lemma, which is a generalized and improved version of a similar lemma appearing in~\cite{KdW04}.

\begin{lemma}[Lemma 2 of \cite{KdW04}]
\label{lemma:kdw}
Let $q \ge 2$ be an integer and let $\cC\colon\Bits^k \to \Sigma^n$ be a code. Let $\cH_1, \ldots , \cH_k$ be $q$-uniform matchings over $[n]$ such that for each $i \in [k]$, we have $\abs{\cH_i} \ge \eps \delta n/q^2$, and suppose that for each $C \in \cH_i$, there exists a function $f_C\colon\Sigma^q \to \Bits$ such that $\Pr_{b \gets \Bits^k}[b_i = f_C(\cC(b)\vert_C)] \geq \frac{1}{2} + \frac{\eps}{2}$.

Then, there exists a binary code $\cC'\colon\Bits^k \to \Bits^{n'}$ with $n' \le 4n\abs{\Sigma}$ and $q$-uniform matchings $\cH_1', \ldots , \cH_k'$ over $n'$ vertices such that for all $i \in [k]$, we have $\abs{\cH_i'} \ge \eps \delta n'/(4q^2\abs{\Sigma})$. Furthermore, for any query set $C \in \cH_i'$, we have that $\Pr_{b \gets \Bits^k}[b_i = \oplus_{v \in C}{\cC'(b)_v}] \geq \frac{1}{2} + \frac{\eps}{2^q\abs{\Sigma}^{q/2}}$.
\end{lemma}

Combining \cref{lemma:kt} and \cref{lemma:kdw}, we immediately obtain \cref{lemma:gen-normal-form}; \cref{thm:main-gen-alpha} then follows by applying \cref{mthm:main}. Thus, it remains to prove \cref{lemma:kdw}. In what follows, we use conventional notations of Boolean analysis from~\cite{OD14}.

\begin{proof}[Proof of \cref{lemma:kdw}]
Consider a natural number $\ell \in \N$ such that $\abs{\Sigma} < 2^\ell \le 2\abs{\Sigma}$, and let $n' \coloneqq n2^{\ell+1}$. Without loss of generality, say that $\Sigma \subseteq \Bits^\ell$. Consider the first-order Reed-Muller encoding $\text{RM}_1\colon\Bits^\ell \to \Bits^{2^{\ell+1}}$ defined as $\text{RM}_1(\sigma) = (\inner{a}{\sigma} + t)_{a \in \Bits^\ell, t \in \Bits}$.\footnote{Here, $\inner{\cdot}{\cdot}$ denotes the pointwise inner product over $\F^{\ell}_2$.} We define our new code $\cC'\colon\Bits^k \to \Bits^{n'}$ as $\cC'(b) \coloneqq (\text{RM}_1(\cC(b)_1), \ldots , \text{RM}_1(\cC(b)_n))$.

Consider any message index $i \in [k]$ and query set $C \in \mc{H}_i$. We are going to find a corresponding query set for $C$ in $\cC'$. Write $C = \{v_1, \ldots , v_q\}$. Arbitrarily extend our function $f_C$ to a function over $(\Bits^\ell)^q$ by setting $f_C(\sigma) = 0$ for $\sigma \in \Bits^\ell \setminus \Sigma$. For any message $b \in \Bits^k$, set $x \coloneqq \cC(b)$. Switching from $\Bits$ to $\Fits$ in the natural way, we find that
\begin{equation*}
    \Pr_{b \gets \Bits^k}[b_i = f_C(\cC(b)|_C)] \geq \frac{1}{2} + \frac{\eps}{2} \iff \underset{b \gets \Fits^k}{\E}[b_i f_C(x_{v_1}, \ldots , x_{v_q})] \ge \eps \ .
\end{equation*}
Consider the Fourier expansion of $f_C$, written as $f_C(y_1, \ldots , y_q) = \sum_{S_1, \ldots , S_q \subseteq [\ell]}{\widehat{f_C}(S_1, \ldots , S_q)\prod_{t=1}^q{\prod_{j \in S_t}{(y_t)_j}}}$. Using the Fourier expansion of $f_C$, the Cauchy-Schwarz inequality, and Parseval's identity, we have
\begin{align*}
    \eps^2 &\le \underset{b \gets \Fits^k}{\E}[b_i f_C(x_{v_1}, \ldots , x_{v_q})]^2 \\
    &= \left(\sum_{S_1, \ldots , S_q \subseteq [\ell]}{\widehat{f_C}(S_1, \ldots , S_q)\underset{b \gets \Fits^k}{\E}\left[b_i \prod_{t=1}^q{\prod_{j \in S_t}{(x_{v_t})_j}}\right]}\right)^2 \\
    &\le \left(\sum_{S_1, \ldots , S_q \subseteq [\ell]}{\widehat{f_C}(S_1, \ldots , S_q)^2}\right) \left(\sum_{S_1, \ldots , S_q \subseteq [\ell]}{\underset{b \gets \Fits^k}{\E}\left[b_i \prod_{t=1}^q{\prod_{j \in S_t}{(x_{v_t})_j}}\right]^2}\right) \\
    &= \left(\underset{y_1, \ldots y_q \gets \Fits^\ell}{\E}[f_C(y_1, \ldots , y_q)^2]\right) \left(\sum_{S_1, \ldots , S_q \subseteq [\ell]}{\underset{b \gets \Fits^k}{\E}\left[b_i \prod_{t=1}^q{\prod_{j \in S_t}{(x_{v_t})_j}}\right]^2}\right) \\
    &= \sum_{S_1, \ldots , S_q \subseteq [\ell]}{\underset{b \gets \Fits^k}{\E}\left[b_i \prod_{t=1}^q{\prod_{j \in S_t}{(x_{v_t})_j}}\right]^2} \\
    &\le 2^{q\ell}\max_{S_1, \ldots , S_q \subseteq [\ell]}\left\{\underset{b \gets \Fits^k}{\E}\left[b_i \prod_{t=1}^q{\prod_{j \in S_t}{(x_{v_t})_j}}\right]^2\right\}
\end{align*}
Thus we can find sets $R_1^C, \ldots , R_q^C \subseteq [\ell]$ and bit $t_C \in \Bits$ such that
\begin{equation*}
    (-1)^{t_C}\underset{b \gets \Fits^k}{\E}\left[b_i \prod_{t=1}^q{\prod_{j \in S_t}{(x_{v_t})_j}}\right] \ge \frac{\eps}{2^{q\ell/2}} \ge \frac{\eps}{2^{q-1}\abs{\Sigma}^{q/2}} \ .
\end{equation*}
Reverting back from $\Fits$ to $\Bits$ in the natural way, the last expression is equivalent to
\begin{equation*}
    \Pr_{b \gets \Bits^k}\left[t_C + \sum_{i=1}^q{\inner{\1_{R_1^C}}{x_{v_i}}} = b_i\right] \geq \frac{1}{2} + \frac{\eps}{2^q\abs{\Sigma}^{q/2}} \ .
\end{equation*}
Thus we can form a new query set $C' \coloneqq \{(v_1, (\1_{R_1^C}, t_C)), (v_2, (\1_{R_2^C}, 0)), \ldots , (v_q, (\1_{R_q^C}, 0))\}$ for $\cC'$ that recovers $b_i$ with probability $1/2 + \eps/(2^q\abs{\Sigma}^{q/2})$. Indeed, this is how we construct our new hypergraphs $\cH_1', \ldots, \cH_k'$. Since we are mapping each query set to a new one, then we see that $\abs{\cH_i} = \abs{\cH_i'} \ge \eps \delta n/q^2 \ge \eps \delta n'/(4q^2\abs{\Sigma})$ for all $i \in [k]$. Furthermore, the query mapping preserves disjointness and size, implying that the new hypergraph is a collection of $k$ $q$-uniform matchings. This finishes the proof.
\end{proof}


\section{Our Proof as a Black-box Reduction to $2$-LDC Lower Bounds}
\label{sec:2-ldc-reduction}


In this appendix, we reinterpret our proof of \cref{mthm:main} in the specific case of \emph{linear} $3$-LDCs by formulating it as a black-box reduction to existing linear $2$-LDC lower bounds. Because we are reinterpreting the proof, we will assume familiarity with the proof in \cref{sec:proof,sec:3xor}.
Formally, we show that our proof of \cref{mthm:main} in fact provides the following transformation: given a \emph{linear} $3$-LDC $\cC$, we produce $2$ different linear codes $\cC_2$ and $\cC_3$ corresponding to the $2$-XOR instance $g_b$ and $3$-XOR instance $f_b$ from \cref{sec:proof}, with the guarantee that at least one of these codes is a linear $2$-LDC.
We note that unlike \cref{mthm:main}, this reduction-based proof will only apply to \emph{linear} $3$-LDCs. However, in this case we will obtain slightly better dependencies on $\log n$, $\eps$, and $\delta$ than that in \cref{mthm:main}; this comes entirely from the fact that $2$-LDC lower bounds for linear codes have slightly better dependencies on $\eps$ and $\delta$ than $2$-LDC lower bounds for general, nonlinear codes.

Our transformation naturally produces objects that are formally not quite linear $2$-LDCs, which we call ``weak LDCs'', defined below.
\begin{definition}[Linear weak LDC]
\label{def:linear-wldc}
Given a code $\cC\colon\Bits^k \to \Bits^n$, we say that $\cC$ is a linear \emph{$(q,\delta)$-weakly locally decodable code} (or, $(q,\delta)$-wLDC) if $\cC$ is a linear code and there are $q$-uniform hypergraph matchings $\cH_1, \ldots , \cH_k$ over $[n]$ such that \begin{inparaenum}[(1)] \item $ \sum_{i=1}^k{\abs{\cH_i}} \geq \delta n k$ for any $i \in [k]$, and \item $C \in \cH_i$, we have that $\bigoplus_{v \in C}{\cC(b)_v} = b_i$ for all messages $b \in \Bits^k$.\end{inparaenum}
\end{definition}
We note that we work with weak LDCs solely for notational convenience, as it is straightforward to observe that they are equivalent to LDCs, up to constant factors in parameters. Indeed, the difference between a weak LDC and a true LDC is that the weak LDC only requires that $\sum_{i = 1}^k \abs{\cH_i} \geq \delta nk$, rather than the stronger condition that $\abs{\cH_i} \geq \delta n$ for all $i \in [k]$. So, by removing all hypergraphs $\cH_i$ with $\abs{\cH_i} \leq \delta n/2$ and setting the corresponding $b_i$'s to $0$, we obtain a new code $\cC' \colon \Bits^{k'} \to \Bits^n$ where $k' \geq \delta k$ and $\abs{\cH_i} \geq \delta n/2$ for all $i \in [k']$. 

Regardless, we note that the linear $2$-LDC lower bound of \cite{GKST06}, which here we will use as a black-box, holds for linear weak $2$-LDCs as well.
\begin{lemma}[Lemma 3.3 of \cite{GKST06}]
\label{lemma:gkst-lemma}
Any linear $(2,\delta)$-wLDC $\cC\colon\Bits^k \to \Bits^n$ satisfies $n \ge 2^{\delta k}$.
\end{lemma}

As the main theorem in this section, we will prove the following theorem.
\begin{theorem}
\label{thm:linreduction}
Let $\cC \colon \Bits^k \to \Bits^n$ be a linear $(3, \delta)$-wLDC, and let $d \in \N$. Then, there are codes $\cC_2 \colon \Bits^{k_2} \to \Bits^{n}$ and $\cC_3 \colon \Bits^{k_3} \to \Bits^N$ such that either $\cC_2$ is a linear $(2, \Omega(\delta \cdot \frac{d}{d + k}))$-wLDC or $\cC_3$ is a linear $(2, \Omega(\delta^2/d))$-wLDC, where $k_2, k_3 \geq k/2$, $N = {2n \choose \ell}$ and $\ell = \sqrt{n/k}/c$, where $c$ is an absolute constant.
\end{theorem}
We note that by applying \cref{lemma:gkst-lemma} twice, we immediately obtain the following corollary.

\begin{corollary}
\label{cor:linlb}
Let $\cC \colon \Bits^k \to \Bits^n$ be a $(3, \delta)$-linear LDC. Then, $n \ge \Omega\left(\frac{\delta^6 k^3}{\log^4{k}}\right)$.
\end{corollary}
\begin{proof}
Apply \cref{thm:linreduction} with $d = c \log_2 n/\delta$ for a sufficiently large constant $c$. If $k \leq d$, then we are done, so suppose that $k \geq d$. If $\cC_2$ is a linear weak $(2, \Omega(\delta \cdot \frac{d}{d + k}))$-LDC, then by \cref{lemma:gkst-lemma} we conclude that $\log_2 n \geq \Omega(\delta dk/(k+d)) \geq \Omega(\delta d)$, as $k + d \leq 2k$. As $d = c \log_2 n/\delta$ for a sufficiently large constant $c$, this is a contradiction. 

It thus cannot be the case that $\cC_2$ is a linear weak $(2, \Omega(\delta \cdot \frac{d}{d + k}))$-LDC, and therefore it must be the case that $\cC_3$ is a linear weak $(2, \Omega(\delta^2/d))$-LDC. By \cref{lemma:gkst-lemma}, this implies that $O(\sqrt{n/k} \log n) \geq \ell \log_2 n \geq \Omega(\delta^2/d \cdot k)$, and therefore we conclude that $n \geq \Omega(\delta^6 k^3/\log^4 n)$. Finally, we have $\log_2 n = \Theta(\log k)$ or else \cref{cor:linlb} trivially holds, and so this finishes the proof.
\end{proof}

We now prove \cref{thm:linreduction}.
\begin{proof}[Proof of \cref{thm:linreduction}]
Let $\cC \colon \Bits^k \to \Bits^n$ be a linear $(3,\delta)$-wLDC, so that there exist $3$-uniform hypergraph matchings $\cH_1, \dots, \cH_k$ such that $\sum_{i = 1}^k \abs{\cH_i} \geq \delta n k$, and for every $i \in [k]$ and $C \in \cH_i$, it holds that $\bigoplus_{v \in C}{\cC(b)_v} = b_i$ for all $b \in \Bits^k$.

We now define the codes $\cC_2$ and $\cC_3$. Let $G_1, \dots, G_k, \cH'_1, \dots, \cH'_k$ denote the output of the hypergraph decomposition algorithm \cref{lem:decomp} applied with the parameter $d$ chosen in the statement of \cref{thm:linreduction}.

\parhead{Constructing $\cC_2$.} Let $L_2 \subseteq [k]$ be a subset of size $\abs{L_2} \geq k/2$ to be specified later. We let $\cC_2 \colon \Bits^{L_2} \to \Bits^n$ be the code that encodes a message $b' \in \Bits^{L_2}$ as $\cC(b)$, where $b$ is obtained by padding $b'$ with $0$'s to obtain $b \in \Bits^k$. Formally, $\cC_2(b') \defeq \cC(b)$, where $b \in \Bits^k$ satisfies $b_i = b'_i$ for all $i \in L_2$ and $b_j = 0$ otherwise.

We will now show that if  $\sum_{i = 1}^k \abs{G_i} \geq \delta n k /2$, then there exists a set $L_2 \subseteq [k]$ of size $\abs{L_2} \geq k/2$ such that $\cC_2$ is a linear $(2, \Omega(\delta \cdot \frac{d}{d + k}))$-wLDC. Recall that each $G_i$ is a bipartite matching on $[n] \times P$, where $P = \{ p = (u,v) : \deg_{\cH}(p) \geq d\}$, where $\cH = \cup_{i = 1}^k \cH_i$. First, by duplicating elements of the set $P$, we can furthermore assume that each $p \in P$ appears not just in at least $d$ edges across all $G_i$'s, but also in at most $2d$ edges. Partition $[k]$ into $L_2 \cup R_2$, and without loss of generality assume $\abs{L_2} \geq k/2$. For $i \in L_2$, let $G'_i$ denote the graph on $n$ vertices with edges $E_i = \{(u,v) : \exists p \in P, j \in R_2, (u,p) \in G_i, (v,p) \in G_j\}$. Observe that $\sum_{i \in L_2} \abs{G'_i} \geq \Omega(\delta n k d)$ in expectation over a random partition $L_2 \cup R_2$, and hence there exists such a partition $L_2 \cup R_2$ with $\sum_{i \in L_2} \abs{G'_i} \geq \Omega(\delta n k d)$. 

Next, we observe that for any vertex $u \in [n]$ and $i \in L_2$, $u$ has degree at most $2d + k$ in $G'_i$. Indeed, since the $G_i$'s are matchings and each $p$ appears in at most $2d$ edges, it follows that for each $u$, there are at most $2d$ edges $(u,v)$ in $G'_i$ formed from the edge $(u,p)$ in $G_i$. Second, for each $v$, there are at most $k$ edges $(u,v)$ in $G'_i$, as these can only be formed from the edges $(v, p)$ in $G_j$, for $j \in R_2$, and each $G_j$ is matching so there is at most one edge per choice of $j \in R_2$. Hence, each $G'_i$ has a matching $M'_i$ of size at least $\Omega(\abs{G'_i}/(d + k))$, and so $\sum_{i = 1}^k \abs{M'_i} \geq \Omega(\delta n k \cdot \frac{d}{d + k})$.

Finally, for each $i \in L_2$ and each edge $(u,v) \in M'_i$, it holds that $\cC_2(b')_{u} \oplus \cC_2(b')_{v} = b'_i$. Indeed, this is because $\cC(b)$ satisfies $\cC(b)_u \oplus \cC(b)_p = b_i$ and $\cC(b)_v \oplus \cC(b)_p = b_j = 0$, where $p \in P$ is the shared pair used to add $(u,v)$ to $G'_i$ in the definition, $j \in R_2$, and $(u,p) \in G_i, (v,p) \in G_j$. We have thus shown that if $\sum_{i =1}^k \abs{G_i} \geq \delta nk/2$, then $\cC_2$ is a linear $(2, \Omega(\delta \cdot \frac{d}{d + k}))$-wLDC.

\parhead{Constructing $\cC_3$.} Let $L_3 \subseteq [k]$ be a subset of size $\abs{L_3} \geq k/2$ to be specified later. Let $\ell = \sqrt{n/k}/c$ for a sufficiently large constant $c$, and identify $N = {2n \choose \ell}$ with the collection of sets ${[n] \times [2] \choose \ell}$.
We let $\cC_3 \colon \Bits^{L_3} \to \Bits^N$ be the code that encodes a message $b' \in \Bits^{L_3}$ with the string $\cC_3(b')$, where the $S$-th entry, for $S \in {[n] \times [2] \choose \ell}$, is 
$$\cC_3(b')_S \defeq (\bigoplus_{u^{(1)} \in S} \cC(b)_u) \oplus (\bigoplus_{v^{(2)} \in S} \cC(b)_v) \ , $$ where $b \in \Bits^k$ satisfies $b_i = b'_i$ for all $i \in L_3$ and $b_j = 0$ otherwise.

We now argue that if $\sum_{i = 1}^k \abs{\cH'_i} \geq \delta n k/2$, then there exists a set $L_3 \subseteq [k]$ of size $\abs{L_3} \geq k/2$ such that $\cC_3$ is a linear $(2, \Omega(\delta^2/d))$-wLDC.
Recall that each $\cH'_i$ is a $3$-uniform hypergraph matching on $n$ vertices, where $\deg_{\cH'}(\{u,v\}) \leq d$ for all $u,v \in [n]$, where $\cH' \defeq \cup_{i = 1}^k \cH'_i$. Partition $[k]$ into $L_3 \cup R_3$, and without loss of generality assume $\abs{L_3} \geq k/2$. Following \cref{sec:3xor}, we set $\ell = \sqrt{n/k}/c$ for a sufficiently large constant $c$ and let $B_i \in \R^{N \times N}$ for $i \in L_3$ be the matrices defined in \cref{def:kikuchi-matrix}. 

Let $G''_i$ denote the graph with adjacency matrix $B_i$, i.e., for $S, T \in [N]$, we have $(S,T)$ as an edge in $G''_i$ if $B_i(S,T) \ne 0$. By \cref{lem:degbound}, the max degree of any vertex in $G''_i$ is at most $2d$. Hence, $G''_i$ contains a matching $M''_i$ where $\abs{M''_i} \geq \Omega(\abs{G''_i}/d)$. Now, since $\abs{\cH'} \ge \delta nk/2$, then by double counting, the number of clauses $C_1, C_2 \in \cH'$ with $\abs{C_1 \cap C_2} \ge 1$ is at least $\Omega(\delta^2nk^2)$. Thus, by picking a random partition and using \cref{lem:rowpruning}, we find that $\sum_{i = 1}^k \abs{G''_i} \geq \Omega(D \delta^2 n k^2)$ in expectation, where $D = 2 {2n - \ell \choose \ell - 4}$, and hence there is a partition $L_3 \cup R_3$ achieving this. By applying \cref{fact:binomialratio}, we see that $D/N \geq \Omega(\ell^2/n^2)$, and so we have $\sum_{i = 1}^k \abs{M''_i} \geq \Omega(\delta^2 N k/d)$, using that $\ell = \sqrt{n/k}/c$.

It is now straightforward to observe that, for each $i \in L_3$ and $(S,T) \in M''_i$, it holds that $b'_i = \cC_3(b')_S \oplus \cC_3(b')_T$; indeed, this is because $\cC_3(b')_S \oplus \cC_3(b')_T = \cC(b)_S \oplus \cC(b)_T = b_i \oplus b_j = b'_i$, as $b'_i = b_i$ and $b_j = 0$ because $j \in R_2$. We have thus shown that if $\sum_{i = 1}^k \abs{\cH'_i} \geq \delta n k /2$, then $\cC_3$ is a linear $(2, \Omega(\delta^2/d))$-wLDC.

\medskip
By \cref{lem:decomp}, we thus have that either $\sum_{i = 1}^k \abs{G_i} \geq \delta n k/2$ or $\sum_{i = 1}^k \abs{\cH'_i} \geq \delta n k /2$. Hence, at least one of $\cC_2$ and $\cC_3$ must have the desired property, which finishes the proof.
\end{proof}
\begin{remark}[A note on the linearity of $\cC$]
\label{rem:linearity}
In \cref{thm:linreduction}, we assumed that the code $\cC$ was linear. The reason that this assumption is necessary is because of the following. The constraints used to locally decode $\cC_2$ and $\cC_3$ are obtained by XORing two clauses $C_1$ and $C_2$ in the original set of local constraints defining $\cC$. We then observe that by using $C_1 \oplus C_2$, we can decode, e.g., $b_i \oplus b_j$, and so by setting $\sim k/2$ of the $b_j$'s to be hardcoded to $0$, we have many constraints to recover $b_i$. The issue for nonlinear codes is that this ``hardcoding'' procedure does not work, as even though we can set $b_j$ to be $0$, the individual constraints $C_1$ and $C_2$ are only guaranteed to decode $b_i$ and $b_j$, respectively, \emph{in expectation} over a random choice of $b \in \Bits^k$. Thus, when we hardcode some bits, we are no longer guaranteed that the derived constraint $C_1 \oplus C_2$ decodes $b_i$ in expectation over the remaining ``free'' bits $b_i$ for $i \in L$.
\end{remark}

\end{document}